\documentclass[11pt]{article}

\usepackage[utf8]{inputenc}
\usepackage[T1]{fontenc}
\usepackage{amsmath,amsthm}
\usepackage{newtxtext}
\usepackage[amssymbols]{newtxmath} % Times font with full math support
\usepackage{graphicx}
\usepackage{xcolor}
\usepackage{booktabs}
\usepackage{algorithm}
\usepackage{algpseudocode}
\usepackage{tikz}
\usetikzlibrary{arrows.meta,positioning,shapes.geometric}
\usepackage[colorlinks=true,linkcolor=blue!70!black,citecolor=blue!70!black,urlcolor=blue!70!black]{hyperref}
\usepackage{cite}
\usepackage[margin=1in]{geometry}
\usepackage{microtype}
\usepackage{caption}
\newtheorem{theorem}{Theorem}[section]

\newtheorem{definition}{Definition}[section]
\newtheorem{remark}{Remark}[section]

\newcommand{\RD}{R(D)}
\newcommand{\Hp}{H(p)}
\newcommand{\HD}{H(D)}
\newcommand{\Ber}{\mathrm{Bernoulli}}
\newcommand{\E}{\mathbb{E}}
\newcommand{\Var}{\mathrm{Var}}
\newcommand{\Prob}{\mathbb{P}}

\newcommand{\calX}{\mathcal{X}}
\newcommand{\calXhat}{\hat{\mathcal{X}}}
\newcommand{\Qinv}{Q^{-1}}

\title{\LARGE Finite Block Length Rate-Distortion Theory for the\\Bernoulli Source with Hamming Distortion: A Tutorial}
\author{Bhaskar Krishnamachari\\[4pt]
{\normalsize Viterbi School of Engineering}\\
{\normalsize University of Southern California}\\[2pt]
{\normalsize \texttt{bkrishna@usc.edu}}}
\date{\normalsize February 27, 2026}

\begin{document}
\maketitle

% ========================================================================
% ABSTRACT
% ========================================================================
\begin{abstract}
Lossy data compression lies at the heart of modern communication and storage systems.
Shannon's rate-distortion theory provides the fundamental limit on how much a source
can be compressed at a given fidelity, but it assumes infinitely long block lengths
that are never realized in practice.
We present a self-contained tutorial on rate-distortion theory for the simplest
non-trivial source: a Bernoulli$(p)$ sequence with Hamming distortion.
We derive the classical rate-distortion function $\RD = \Hp - \HD$ from first
principles, illustrate its computation via the Blahut-Arimoto algorithm, and then
develop the finite block length refinements that characterize how the minimum
achievable rate approaches the Shannon limit as the block length $n$ grows.
The central quantity in this refinement is the \emph{rate-distortion dispersion}
$V(D)$, which governs the $O(1/\sqrt{n})$ penalty for operating at finite block
lengths.
We accompany all theoretical developments with numerical examples and figures
generated by accompanying Python scripts.
\end{abstract}

% ========================================================================
% SECTION 1: INTRODUCTION
% ========================================================================
\section{Introduction}
\label{sec:introduction}

The theory of lossy data compression traces its origins to Claude Shannon's
landmark 1948 paper~\cite{shannon1948}, which established that every source has
a well-defined minimum description rate for any prescribed level of
distortion.
This result, made precise in Shannon's 1959 coding theorem for sources with a
fidelity criterion~\cite{shannon1959}, was remarkable for a reason that is easy
to overlook today: it demonstrated that a single, clean mathematical function,
the rate-distortion function $\RD$, separates the achievable from the impossible, no
matter how clever the compression scheme.

However, this elegant theory rests on a crucial idealization.
The rate-distortion function $\RD$ is an \emph{asymptotic} quantity: it describes
the minimum rate achievable when the block length $n$ tends to infinity.
Real communication and storage systems must operate with finite memory, finite
latency, and finite computational resources.
A natural and practically important question therefore arises: \emph{how much
extra rate do we need when the block length is finite?}

To build intuition, consider the simplest possible lossy compression scheme.
Suppose we have a $\Ber(0.5)$ source that produces sequences of $n = 2$ bits.
There are four possible source sequences: $00$, $01$, $10$, and $11$.
We wish to compress each sequence to just $1$ bit, so our codebook has
$M = 2$ entries and the rate is $R = \frac{1}{2}\log_2 2 = 0.5$ bits per
source symbol.
One natural code assigns codeword $0$ to the pair $\{00, 01\}$ and codeword $1$
to the pair $\{10, 11\}$: the encoder simply keeps the first bit and discards
the second.
The decoder reconstructs $00$ from codeword $0$ and $11$ from codeword $1$.
What distortion does this code achieve?
Under Hamming distortion (which counts the fraction of disagreeing positions),
the four source sequences yield distortions
$d(00, 00) = 0$, $d(01, 00) = 1/2$, $d(10, 11) = 1/2$, and $d(11, 11) = 0$.
Since the source is fair, each sequence is equally likely, so the average
distortion is $(0 + 1/2 + 1/2 + 0)/4 = 1/4$.
We have thus constructed a concrete code operating at rate $R = 0.5$ and
distortion $D = 0.25$.

Is this the best we can do at rate $0.5$?
Shannon's rate-distortion theory provides the answer.
One of the landmark results of information theory, derived in detail in
Section~\ref{sec:bernoulli_rd}, is that for a $\Ber(p)$ source with Hamming
distortion, the minimum achievable rate at distortion level $D$ is
\begin{equation}
\label{eq:rd_intro}
\RD = \Hp - \HD, \qquad 0 \leq D \leq \min(p, 1-p),
\end{equation}
where $\Hp = -p\log_2 p - (1-p)\log_2(1-p)$ is the \emph{binary entropy
function}, which measures the inherent uncertainty of the source.
For a fair coin ($p = 0.5$), we have $H(0.5) = 1$ bit, so the formula
simplifies to $\RD = 1 - \HD$.
Evaluating at $D = 0.25$ gives $R(0.25) = 1 - H(0.25) \approx 0.189$ bits
per symbol.
Since our simple code operates at rate $0.5$, well above the Shannon limit of
$0.189$, there is considerable room for improvement, but only in the limit of
large block lengths.
For $n = 2$, it turns out, we cannot do better.
To see why, note that any code with $M = 2$ assigns each of the four source
sequences to one of two reconstructions $\hat{x}_1, \hat{x}_2 \in \{0,1\}^2$
(which need not belong to the assigned group).
In $\{0,1\}^2$, each reconstruction can match at most one sequence exactly
(itself, at Hamming distance~$0$); every other sequence has Hamming
distance at least~$1$, contributing per-symbol distortion at least~$1/2$.
With only two reconstructions, at most two of the four sequences can achieve
distortion~$0$, so at least two sequences contribute distortion~$\geq 1/2$.
Since all four are equally likely, the average distortion is at least
$(0 + 0 + 1/2 + 1/2)/4 = 1/4$, which is exactly what our code achieves.
The gap between this finite-$n$ optimum of $D = 0.25$ at rate $0.5$ and the
Shannon limit of $R(0.25) \approx 0.189$ at the same distortion illustrates a
fundamental phenomenon: short codes pay a rate penalty compared to the
asymptotic limit.
Quantifying this penalty precisely is the goal of finite block length theory.

Over the past two decades, a precise answer to this question has emerged through
the work of Strassen~\cite{strassen1962},
Ingber and Kochman~\cite{ingber2011},
Kostina and Verd\'{u}~\cite{kostina2012}, and others.
To state their result, we need two preliminary ideas.

The first is a \emph{distortion measure}: a function $d(x, \hat{x})$ that
quantifies how ``far'' a reconstruction symbol $\hat{x}$ is from the original
source symbol $x$.
For binary data, the simplest and most natural choice is the \emph{Hamming
distortion}, which equals $1$ when $x \neq \hat{x}$ and $0$ when $x = \hat{x}$.
When we compress a length-$n$ source sequence $X^n = (X_1, \ldots, X_n)$ into
a reconstruction $\hat{X}^n = (\hat{X}_1, \ldots, \hat{X}_n)$, the overall
quality is measured by the \emph{per-symbol distortion}
$\frac{1}{n}\sum_{i=1}^{n} d(X_i, \hat{X}_i)$,
which under Hamming distortion is simply the fraction of positions where the
source and reconstruction disagree (the bit error rate).

The second idea is that we must refine what ``achievable'' means at finite
block length.
When $n$ is finite, no code can guarantee that \emph{every} source sequence is
reproduced within distortion $D$.
The source sequence $X^n$ is random, and some realizations are inherently harder
to compress than others.
We therefore allow a small probability of failure: we require only that the
distortion exceeds $D$ with probability at most $\varepsilon$.
More precisely, we say a code is $(n, D, \varepsilon)$-achievable at rate $R$
if
\begin{equation}
\label{eq:excess_distortion_intro}
\Prob\!\left(\frac{1}{n}\sum_{i=1}^{n} d(X_i, \hat{X}_i) > D\right) \leq \varepsilon,
\end{equation}
where the probability is over the randomness of the source sequence $X^n$.
The quantity $\varepsilon$ is called the \emph{excess-distortion probability}.
With this formulation, the minimum achievable rate $R(n, D, \varepsilon)$
depends on three parameters: the block length $n$, the target distortion $D$,
and the tolerated failure probability $\varepsilon$.

The key insight of the finite block length theory is that $R(n, D, \varepsilon)$
admits a clean asymptotic expansion:
\begin{equation}
\label{eq:normal_approx_intro}
R(n, D, \varepsilon) \approx \RD + \sqrt{\frac{V(D)}{n}}\, \Qinv(\varepsilon),
\end{equation}
where $V(D)$ is the \emph{rate-distortion dispersion}, a quantity that captures
how variable the compression difficulty is across source symbols, and
$\Qinv(\varepsilon)$ is the inverse of the Gaussian $Q$-function.
From an engineering standpoint, (\ref{eq:normal_approx_intro}) reveals that the
penalty for finite block length decays as $1/\sqrt{n}$, a rate that is neither
negligibly fast nor prohibitively slow.

In this tutorial, we develop the entire story from first principles for the simplest
non-trivial setting: a $\Ber(p)$ source with Hamming distortion.
We have chosen this source for three reasons.
First, the binary symmetric source is the discrete analogue of the Gaussian source:
it is the canonical ``textbook'' example against which all intuitions are calibrated.
Second, every quantity of interest ($\RD$, the optimal test channel, the $d$-tilted
information, and the dispersion $V(D)$) admits a clean closed-form expression.
Third, despite its simplicity, the Bernoulli source reveals the full structure of
the finite block length theory, including the role of source dispersion as a
second-order characterization of compression difficulty.

The key contributions of this tutorial are:
\begin{enumerate}
    \item A self-contained derivation of the rate-distortion function $\RD = \Hp - \HD$
          for the $\Ber(p)$ source, accessible to readers with minimal probability background.
    \item A detailed treatment of the Blahut-Arimoto algorithm, including explicit
          $2 \times 2$ matrix computations and convergence analysis.
    \item A development of finite block length rate-distortion theory, including the
          $d$-tilted information, rate-distortion dispersion, and the normal approximation.
    \item Accompanying Python scripts that reproduce all numerical results and figures.
\end{enumerate}

The remainder of this paper is organized as follows.
Section~\ref{sec:foundations} reviews the probability and information-theoretic
foundations.
Section~\ref{sec:rd_problem} formulates the rate-distortion problem.
Section~\ref{sec:bernoulli_rd} derives the rate-distortion function for the
Bernoulli source.
Section~\ref{sec:blahut_arimoto} presents the Blahut-Arimoto algorithm.
Section~\ref{sec:finite_blocklength} develops the finite block length theory.
Section~\ref{sec:numerical} presents comprehensive numerical explorations.
Finally, Section~\ref{sec:conclusion} concludes with a discussion of open
problems and further reading.

% ========================================================================
% SECTION 2: PROBABILITY AND INFORMATION FOUNDATIONS
% ========================================================================
\section{Probability and Information Foundations}
\label{sec:foundations}

In this section, we review the essential probability and information-theoretic
concepts that underpin rate-distortion theory.
We aim for an intuitive, example-driven development; readers seeking formal
generality may consult Cover and Thomas~\cite{cover2006}.

% ---- 2.1 ----
\subsection{Random Variables and Probability}
\label{subsec:rv}

Consider a coin that lands heads with probability $p$ and tails with
probability $1 - p$, where $0 \leq p \leq 1$.
If we encode heads as $1$ and tails as $0$, a single coin flip is described
by a \emph{Bernoulli random variable} $X$ with
\begin{equation}
\Prob(X = 1) = p, \qquad \Prob(X = 0) = 1 - p.
\end{equation}
We write $X \sim \Ber(p)$ and refer to $p$ as the \emph{bias} of the source.
A fair coin corresponds to $p = 1/2$, while a biased coin has $p \neq 1/2$.

When we flip the coin $n$ times independently, we obtain a \emph{sequence}
$X^n = (X_1, X_2, \ldots, X_n)$ of independent and identically distributed
(i.i.d.) Bernoulli random variables.
This sequence is the object we wish to compress.

% ---- 2.2 ----
\subsection{Entropy}
\label{subsec:entropy}

How much ``surprise'' does a single coin flip carry?
If the coin always lands heads ($p = 1$), there is no surprise at all; we know
the outcome in advance.
If the coin is fair ($p = 1/2$), each flip is maximally uncertain.
Shannon formalized this intuition through the \emph{entropy} of a random variable.

\begin{definition}[Binary Entropy]
\label{def:binary_entropy}
The \emph{binary entropy function} is defined as
\begin{equation}
\label{eq:binary_entropy}
H(p) = -p \log_2 p - (1-p) \log_2 (1-p),
\end{equation}
with the convention that $0 \log_2 0 = 0$.
\end{definition}

The entropy $\Hp$ measures the average surprise, in bits, of a single draw
from a $\Ber(p)$ source.
It is zero when $p = 0$ or $p = 1$ (no uncertainty) and achieves its maximum
value of $1$ bit when $p = 1/2$ (maximum uncertainty).
Figure~\ref{fig:binary_entropy} illustrates this behavior.

\begin{figure}[htbp]
    \centering
    \includegraphics[width=0.75\textwidth]{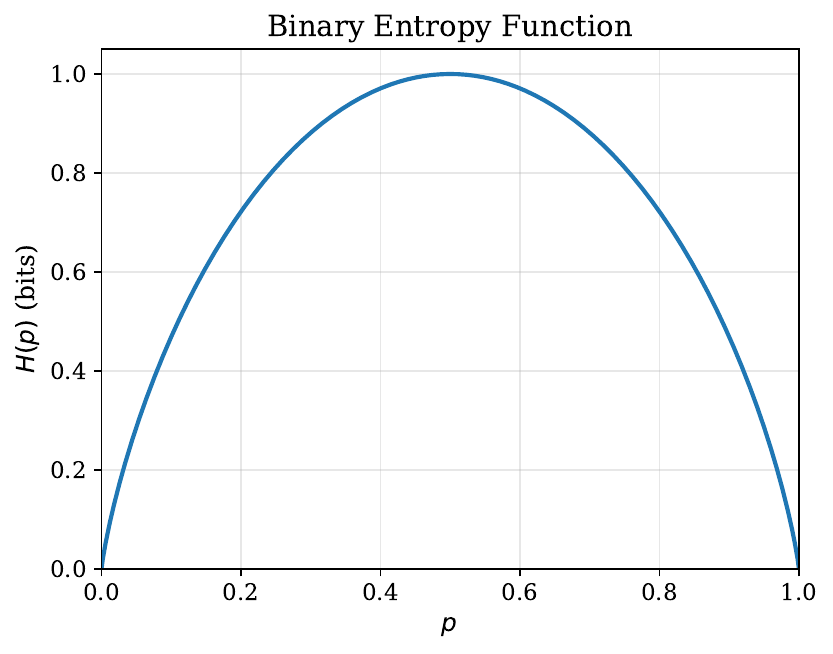}
    \caption{The binary entropy function $\Hp$ versus the source bias $p$.
    The entropy is maximized at $p = 1/2$, where each bit carries one full bit
    of information, and vanishes at $p \in \{0, 1\}$, where the source is
    deterministic.}
    \label{fig:binary_entropy}
\end{figure}

More generally, for a discrete random variable $X$ taking values in a finite
alphabet $\calX$ with probability mass function $p_X(x)$, the entropy is
\begin{equation}
\label{eq:entropy_general}
H(X) = -\sum_{x \in \calX} p_X(x) \log_2 p_X(x).
\end{equation}
The entropy quantifies the minimum average number of bits needed to losslessly
represent $X$.

% ---- 2.3 ----
\subsection{Sequences and Typical Sequences}
\label{subsec:typical}

Consider a sequence $x^n = (x_1, \ldots, x_n)$ produced by a $\Ber(p)$ source.
The \emph{type} of $x^n$ is the empirical distribution of symbols: if $x^n$
contains $k$ ones, its type is $k/n$.
For large $n$, the law of large numbers guarantees that the fraction of ones
concentrates around $p$.

A sequence is called \emph{typical} if its empirical statistics are close to
the true source distribution.
The \emph{asymptotic equipartition property} (AEP) states that with high
probability, a sequence drawn from a $\Ber(p)$ source satisfies
\begin{equation}
-\frac{1}{n} \log_2 p_{X^n}(X^n) \approx \Hp.
\end{equation}
Intuitively, there are approximately $2^{n\Hp}$ typical sequences, and they
account for almost all of the probability mass.
This observation is the foundation of both lossless and lossy compression:
loosely speaking, it suggests that we could represent all likely source sequences
by mapping them to just these $2^{n\Hp}$ typical sequences, each of which could
be identified using only $n \cdot \Hp$ bits.
In other words, each source symbol that nominally requires one bit to describe
can, on average, be compressed down to $\Hp$ bits.

% ---- 2.4 ----
\subsection{Mutual Information}
\label{subsec:mutual_info}

When we compress a source $X$ into a reconstruction $\hat{X}$, some information
about $X$ is preserved and some is lost.
The \emph{mutual information} $I(X; \hat{X})$ quantifies how much information
$\hat{X}$ retains about $X$.

\begin{definition}[Mutual Information]
\label{def:mutual_info}
For jointly distributed random variables $(X, \hat{X})$ with joint distribution
$p_{X,\hat{X}}(x, \hat{x})$, the mutual information is
\begin{equation}
\label{eq:mutual_info}
I(X; \hat{X}) = \sum_{x, \hat{x}} p_{X,\hat{X}}(x, \hat{x})
\log_2 \frac{p_{\hat{X}|X}(\hat{x}|x)}{p_{\hat{X}}(\hat{x})}.
\end{equation}
\end{definition}

An equivalent and often more intuitive expression is
\begin{equation}
\label{eq:mi_entropy}
I(X; \hat{X}) = H(X) - H(X | \hat{X}),
\end{equation}
where $H(X|\hat{X})$ is the conditional entropy, that is, the residual uncertainty
about $X$ after observing $\hat{X}$.
The mutual information is always non-negative, and equals zero if and only if
$X$ and $\hat{X}$ are independent.
It equals $H(X)$ when $\hat{X}$ determines $X$ perfectly.

% ========================================================================
% SECTION 3: THE RATE-DISTORTION PROBLEM
% ========================================================================
\section{The Rate-Distortion Problem}
\label{sec:rd_problem}

In this section, we formulate the central problem of lossy compression.
We define what it means to compress a source with a prescribed fidelity, and
state Shannon's fundamental theorem that establishes the existence of a minimum
achievable rate.

% ---- 3.1 ----
\subsection{What Is Lossy Compression?}
\label{subsec:lossy}

Consider a source that emits a sequence $X^n = (X_1, \ldots, X_n)$ of $n$
i.i.d.\ $\Ber(p)$ random variables.
A \emph{lossy compression scheme} consists of two mappings:
\begin{itemize}
    \item An \emph{encoder} $f_n : \{0,1\}^n \to \{1, 2, \ldots, M\}$ that maps
          the source sequence to one of $M$ codewords.
    \item A \emph{decoder} $g_n : \{1, 2, \ldots, M\} \to \{0,1\}^n$ that maps
          each codeword index back to a reconstruction sequence
          $\hat{X}^n = g_n(f_n(X^n))$.
\end{itemize}
The set $\mathcal{C} = \{g_n(1), g_n(2), \ldots, g_n(M)\}$ is called the
\emph{codebook}, and each element is a \emph{reproduction sequence}.

The \emph{rate} of the code is
\begin{equation}
\label{eq:rate_def}
R = \frac{1}{n} \log_2 M \quad \text{bits per source symbol}.
\end{equation}
This quantity measures the average number of bits used to represent each source
symbol.
A lower rate means more aggressive compression.

% ---- 3.2 ----
\subsection{Distortion Measures}
\label{subsec:distortion}

To quantify the fidelity of the reconstruction, we need a \emph{distortion measure}.
For binary sequences, the natural choice is the \emph{Hamming distortion}:
\begin{equation}
\label{eq:hamming}
d(x, \hat{x}) =
\begin{cases}
0, & \text{if } x = \hat{x}, \\
1, & \text{if } x \neq \hat{x}.
\end{cases}
\end{equation}
The Hamming distortion simply counts whether a symbol was reproduced correctly.

For a pair of sequences $(x^n, \hat{x}^n)$, the \emph{per-symbol distortion} is
\begin{equation}
\label{eq:block_distortion}
d(x^n, \hat{x}^n) = \frac{1}{n} \sum_{i=1}^{n} d(x_i, \hat{x}_i),
\end{equation}
which equals the fraction of positions where the source and reconstruction
disagree, that is, the \emph{bit error rate}.

% ---- 3.3 ----
\subsection{The Fundamental Question}
\label{subsec:fundamental}

We can now state the central question of rate-distortion theory:

\medskip
\noindent\emph{What is the minimum rate $R$ such that there exists a sequence of
encoder-decoder pairs achieving average distortion at most $D$?}
\medskip

Shannon's remarkable insight was that, in the limit $n \to \infty$, this minimum
rate converges to a ``single-letter'' quantity $R(D)$---an optimization involving
only the distribution of one source symbol and one reproduction symbol, not the
full length-$n$ sequences~\cite{shannon1959}.
We do not need to search over all possible encoder-decoder pairs of all possible
block lengths: the \emph{asymptotic} answer depends only on the source
distribution $p_X$ and the distortion measure $d$.
(For any finite $n$, the actual minimum rate exceeds $R(D)$; quantifying this
gap is the subject of Section~\ref{sec:finite_blocklength}.)

% ---- 3.4 ----
\subsection{The Test Channel and the Rate-Distortion Function}
\label{subsec:test_channel}

The key to Shannon's formulation is the concept of a \emph{test channel}.
Rather than optimizing over encoder-decoder pairs, we optimize over conditional
distributions $p_{\hat{X}|X}(\hat{x}|x)$ that describe a probabilistic mapping
from source symbols to reconstruction symbols.

\begin{definition}[Rate-Distortion Function]
\label{def:rd_function}
The rate-distortion function of a source $X$ with distortion measure $d$ is
\begin{equation}
\label{eq:rd_function}
\RD = \min_{\substack{p_{\hat{X}|X}: \\ \E[d(X, \hat{X})] \leq D}} I(X; \hat{X}),
\end{equation}
where the minimization is over all conditional distributions $p_{\hat{X}|X}$
satisfying the distortion constraint.
\end{definition}

Why do we \emph{minimize} the mutual information?
The mutual information $I(X; \hat{X})$ measures how many bits the reconstruction
$\hat{X}$ reveals about the original source $X$.
Any information that $\hat{X}$ carries about $X$ must have been conveyed by the
encoder, so $I(X; \hat{X})$ is precisely the rate cost of the code.
The distortion constraint $\E[d(X, \hat{X})] \leq D$ separately ensures that the
reconstruction is faithful enough; our goal is then to find the cheapest
test channel (lowest rate) that still keeps distortion within the budget $D$.

The rate-distortion function therefore has a clean operational interpretation:
$\RD$ is the minimum number of bits per source symbol required to describe the
source with average distortion at most $D$.
Rates above $\RD$ are achievable (there exist codes that work), while rates
below $\RD$ are not achievable by any code, regardless of its complexity.

Figure~\ref{fig:test_channel} illustrates the encoder-decoder structure and the
test channel abstraction.

\begin{figure}[htbp]
    \centering
    \begin{tikzpicture}[>=Stealth, node distance=1.4cm, font=\small]
        % Source
        \node[draw, rounded corners, fill=blue!8, minimum width=1.3cm,
              minimum height=0.8cm] (src) {Source};
        % Encoder
        \node[draw, rounded corners, fill=orange!12, minimum width=1.3cm,
              minimum height=0.8cm, right=of src] (enc) {Encoder};
        % Channel
        \node[right=1.1cm of enc] (mid) {};
        % Decoder
        \node[draw, rounded corners, fill=green!10, minimum width=1.3cm,
              minimum height=0.8cm, right=2.2cm of enc] (dec) {Decoder};
        % Output
        \node[right=0.8cm of dec] (out) {$\hat{X}^n$};

        % Arrows
        \draw[->] (src) -- node[above]{\footnotesize $X^n$} (enc);
        \draw[->] (enc) -- node[above]{\footnotesize $nR$ bits} (dec);
        \draw[->] (dec) -- (out);

        % Test channel below
        \node[below=2.0cm of src, xshift=1.0cm] (x) {$X$};
        \node[draw, rounded corners, fill=purple!10, minimum width=2.2cm,
              minimum height=0.8cm, right=0.8cm of x] (tc) {$p_{\hat{X}|X}$};
        \node[right=0.8cm of tc] (xhat) {$\hat{X}$};
        \draw[->] (x) -- (tc);
        \draw[->] (tc) -- (xhat);

        % Labels
        \node[above=0.15cm of enc, xshift=0.3cm, gray] {\footnotesize Operational};
        \node[above=0.15cm of tc, gray] {\footnotesize Test Channel};
    \end{tikzpicture}
    \caption{Top: the operational lossy compression setup with encoder and decoder.
    Bottom: the test channel $p_{\hat{X}|X}$ that abstracts away the codebook
    structure.  The rate-distortion function minimizes mutual information
    $I(X;\hat{X})$ over all test channels satisfying the distortion constraint.}
    \label{fig:test_channel}
\end{figure}
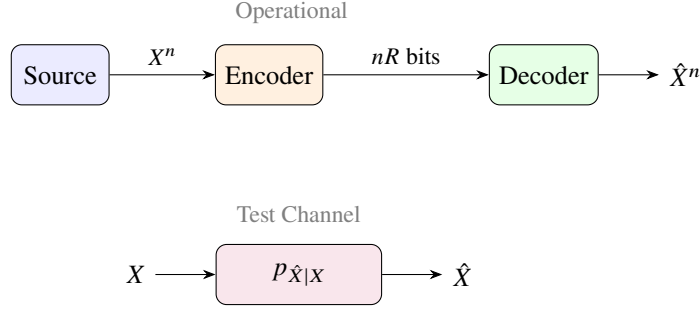

% ========================================================================
% SECTION 4: RATE-DISTORTION FOR THE BERNOULLI SOURCE
% ========================================================================
\section{The Rate-Distortion Function for the Bernoulli Source}
\label{sec:bernoulli_rd}

In this section, we derive the rate-distortion function for the $\Ber(p)$
source with Hamming distortion.
This is perhaps the cleanest closed-form result in all of rate-distortion theory.

% ---- 4.1 ----
\subsection{Setting Up the Optimization}
\label{subsec:rd_setup}

We wish to solve
\begin{equation}
\label{eq:rd_bernoulli_opt}
\RD = \min_{\substack{p_{\hat{X}|X}: \\ \E[d(X, \hat{X})] \leq D}} I(X; \hat{X})
\end{equation}
for $X \sim \Ber(p)$ and Hamming distortion $d(x, \hat{x}) = \mathbf{1}\{x \neq \hat{x}\}$.

Since both the source alphabet $\calX = \{0, 1\}$ and the reproduction alphabet
$\calXhat = \{0, 1\}$ are binary, the test channel $p_{\hat{X}|X}$ is a $2 \times 2$
stochastic matrix with four parameters, of which two are free (each row sums to one).
We can parameterize the test channel as
\begin{equation}
\label{eq:test_channel_matrix}
\begin{pmatrix}
p_{\hat{X}|X}(0|0) & p_{\hat{X}|X}(1|0) \\
p_{\hat{X}|X}(0|1) & p_{\hat{X}|X}(1|1)
\end{pmatrix}
=
\begin{pmatrix}
1 - \alpha & \alpha \\
\beta & 1 - \beta
\end{pmatrix},
\end{equation}
where $\alpha = p_{\hat{X}|X}(1|0)$ is the probability of flipping a $0$ to a $1$,
and $\beta = p_{\hat{X}|X}(0|1)$ is the probability of flipping a $1$ to a $0$.

The expected distortion under this test channel is
\begin{equation}
\E[d(X, \hat{X})] = (1-p)\alpha + p\beta.
\end{equation}
We seek the test channel parameters $(\alpha, \beta)$ that minimize the mutual
information $I(X; \hat{X})$ subject to $(1-p)\alpha + p\beta \leq D$.

% ---- 4.2 ----
\subsection{The Optimal Test Channel}
\label{subsec:optimal_test_channel}

We present two derivations of the optimal test channel.
The first uses the classical Lagrangian/KKT machinery, which exposes the
algebraic structure of the solution.
The second uses an entropy-maximization argument that provides complementary
geometric intuition.

\subsubsection*{Derivation 1: Lagrangian and KKT Conditions}

To solve the optimization~(\ref{eq:rd_bernoulli_opt}), we form the Lagrangian
\begin{equation}
\label{eq:lagrangian}
\mathcal{L} = I(X; \hat{X}) + \lambda \bigl(\E[d(X, \hat{X})] - D\bigr),
\end{equation}
where $\lambda \geq 0$ is the Lagrange multiplier.
We work through the Karush-Kuhn-Tucker (KKT) conditions step by step.

\smallskip
\noindent\textbf{Step 1: Expand the Lagrangian.}
Using the mutual information definition~(\ref{eq:mutual_info}), we write
$I(X;\hat{X})$ as a function of the test channel parameters $\alpha$ and
$\beta$ from~(\ref{eq:test_channel_matrix}).
The marginal reproduction distribution is
\[
p_{\hat{X}}(0) = (1-p)(1-\alpha) + p\beta, \qquad
p_{\hat{X}}(1) = (1-p)\alpha + p(1-\beta).
\]
The mutual information is
\begin{align}
I(X;\hat{X}) &= (1-p)(1-\alpha)\log_2\frac{1-\alpha}{p_{\hat{X}}(0)}
  + (1-p)\alpha\log_2\frac{\alpha}{p_{\hat{X}}(1)} \notag\\
  &\quad + p\beta\log_2\frac{\beta}{p_{\hat{X}}(0)}
  + p(1-\beta)\log_2\frac{1-\beta}{p_{\hat{X}}(1)},
\end{align}
and the distortion is $\E[d(X,\hat{X})] = (1-p)\alpha + p\beta$.
The full Lagrangian is therefore
\[
\mathcal{L}(\alpha,\beta,\lambda) = I(X;\hat{X}) + \lambda\bigl[(1-p)\alpha + p\beta - D\bigr].
\]

\smallskip
\noindent\textbf{Step 2: Stationarity conditions.}
We differentiate $\mathcal{L}$ with respect to $\alpha$ and $\beta$ and set
each derivative to zero.
Differentiating the mutual information with respect to $\alpha$, using
$\partial p_{\hat{X}}(0)/\partial\alpha = -(1-p)$ and
$\partial p_{\hat{X}}(1)/\partial\alpha = (1-p)$, gives (after simplification)
\[
\frac{\partial I}{\partial \alpha}
= (1-p)\log_2\frac{\alpha\, p_{\hat{X}}(0)}{(1-\alpha)\,p_{\hat{X}}(1)}.
\]
Setting $\partial\mathcal{L}/\partial\alpha = 0$ yields
\begin{equation}
\label{eq:kkt_alpha}
\log_2\frac{\alpha\, p_{\hat{X}}(0)}{(1-\alpha)\,p_{\hat{X}}(1)} = -\lambda.
\end{equation}
By an identical calculation with respect to $\beta$:
\begin{equation}
\label{eq:kkt_beta}
\log_2\frac{(1-\beta)\, p_{\hat{X}}(0)}{\beta\,p_{\hat{X}}(1)} = \lambda.
\end{equation}

\smallskip
\noindent\textbf{Step 3: Recognize the exponential (Gibbs) form.}
Equation~(\ref{eq:kkt_alpha}) says
$\alpha/(1-\alpha) = 2^{-\lambda}\,p_{\hat{X}}(1)/p_{\hat{X}}(0)$.
Since $\alpha = p_{\hat{X}|X}(1|0)$ and $1-\alpha = p_{\hat{X}|X}(0|0)$,
this means
\[
\frac{p_{\hat{X}|X}(1|0)}{p_{\hat{X}|X}(0|0)}
= \frac{p_{\hat{X}}(1)}{p_{\hat{X}}(0)}\cdot 2^{-\lambda}.
\]
For source symbol $x = 0$, the distortions are $d(0,0) = 0$ and $d(0,1) = 1$,
so the factor $2^{-\lambda}$ acts as an exponential penalty on the mismatched
reproduction.
Similarly,~(\ref{eq:kkt_beta}) gives
$p_{\hat{X}|X}(0|1)/p_{\hat{X}|X}(1|1) = (p_{\hat{X}}(0)/p_{\hat{X}}(1))
\cdot 2^{-\lambda}$,
again penalizing the mismatch $d(1,0) = 1$.
Both conditions are unified by the \emph{Gibbs form}:
\begin{equation}
\label{eq:gibbs_form}
p^*_{\hat{X}|X}(\hat{x}|x)
= \frac{Q^*(\hat{x})\, 2^{-\lambda\, d(x,\hat{x})}}{Z(x)},
\end{equation}
where $Q^*(\hat{x}) = p^*_{\hat{X}}(\hat{x})$ is the optimal reproduction
distribution (which the test channel must be self-consistent with) and
$Z(x) = \sum_{\hat{x}} Q^*(\hat{x})\,2^{-\lambda\,d(x,\hat{x})}$ is the
normalizing constant ensuring $\sum_{\hat{x}} p^*_{\hat{X}|X}(\hat{x}|x) = 1$.
(The notation $Q^*$ for the reproduction distribution is standard in
rate-distortion theory and should not be confused with the Gaussian $Q$-function
introduced in Section~\ref{subsec:normal_approx}.)

The Gibbs form has a natural interpretation: the optimal test channel takes
the reproduction distribution $Q^*$ and \emph{reweights} each reproduction symbol
$\hat{x}$ by an exponential factor $2^{-\lambda\,d(x,\hat{x})}$.
Symbols close to $x$ (low distortion) keep their weight, while symbols far
from $x$ (high distortion) are exponentially suppressed.
The Lagrange multiplier $\lambda$ controls the severity of this penalty:
larger $\lambda$ (corresponding to tighter distortion constraints) produces
a more concentrated channel.

\smallskip
\noindent\textbf{Step 4: Identify the backward channel.}
The Gibbs form~(\ref{eq:gibbs_form}) reveals a remarkable structure when we
look at the \emph{backward} (reverse) channel via Bayes' rule:
\[
p^*_{X|\hat{X}}(x|\hat{x})
= \frac{p_X(x)\,p^*_{\hat{X}|X}(\hat{x}|x)}{Q^*(\hat{x})}
= \frac{p_X(x)\, 2^{-\lambda\, d(x,\hat{x})}}{Z(x)},
\]
where $Z(x) = \sum_{\hat{x}} Q^*(\hat{x})\,2^{-\lambda\,d(x,\hat{x})}$ is
the same forward normalizer from~(\ref{eq:gibbs_form}).
A crucial subtlety: $Z(x)$ depends on the \emph{source} symbol $x$, not the
reproduction $\hat{x}$.
For Hamming distortion, $d(x,\hat{x}) = 0$ when $x = \hat{x}$ and $1$ when
$x \neq \hat{x}$, so the two normalizers are
\[
Z(0) = Q^*(0) + Q^*(1)\cdot 2^{-\lambda},\qquad
Z(1) = Q^*(0)\cdot 2^{-\lambda} + Q^*(1).
\]
Writing out all four backward-channel entries, grouped by source symbol
(entries sharing the same $x$ share the same denominator):
\begin{alignat*}{2}
&\text{Source }x=0:\quad
& p^*_{X|\hat{X}}(0|0) &= \frac{1-p}{Z(0)},\qquad
  p^*_{X|\hat{X}}(0|1) = \frac{(1-p)\cdot 2^{-\lambda}}{Z(0)},\\[4pt]
&\text{Source }x=1:\quad
& p^*_{X|\hat{X}}(1|0) &= \frac{p\cdot 2^{-\lambda}}{Z(1)},\qquad
  p^*_{X|\hat{X}}(1|1) = \frac{p}{Z(1)}.
\end{alignat*}

\smallskip
\noindent\textbf{Step 5: Complementary slackness and solving for $\lambda$.}
The KKT conditions require \emph{complementary slackness}:
$\lambda\bigl[(1-p)\alpha + p\beta - D\bigr] = 0$.
For $0 < D < \min(p, 1-p)$ the rate is strictly positive, so the distortion
constraint must be active ($\lambda > 0$), giving
$\E[d(X,\hat{X})] = D$.

We now determine $\lambda$ by requiring the backward channel to be a BSC($D$).
Consider the two entries with source symbol $x = 0$ (which share the
denominator $Z(0)$).
Setting $p^*_{X|\hat{X}}(0|0) = 1-D$ gives
\[
Z(0) = \frac{1-p}{1-D}.
\]
Then requiring $p^*_{X|\hat{X}}(0|1) = D$:
\[
\frac{(1-p)\cdot 2^{-\lambda}}{Z(0)}
= \frac{(1-p)\cdot 2^{-\lambda}}{(1-p)/(1-D)}
= (1-D)\cdot 2^{-\lambda} = D,
\]
which gives $2^{-\lambda} = D/(1-D)$, i.e.,
\begin{equation}
\label{eq:lambda_kkt}
\lambda = \log_2\frac{1-D}{D}.
\end{equation}
One can verify the $x = 1$ entries give the same result:
$p^*_{X|\hat{X}}(1|1) = 1-D$ forces $Z(1) = p/(1-D)$, and then
$p^*_{X|\hat{X}}(1|0) = p\cdot 2^{-\lambda}/Z(1) = (1-D)\cdot D/(1-D) = D$.\;\checkmark

Therefore the backward channel is a BSC($D$):
\begin{equation}
\label{eq:backward_channel}
p^*_{X|\hat{X}}(x|\hat{x}) =
\begin{cases}
1-D, & \text{if } x = \hat{x},\\
D, & \text{if } x \neq \hat{x}.
\end{cases}
\end{equation}

\smallskip
\noindent\textbf{Step 6: Derive the reproduction distribution and forward channel.}
From the backward channel BSC($D$) and Bayes' rule, the reproduction
distribution must satisfy
$p_X(x) = \sum_{\hat{x}} Q^*(\hat{x})\,p^*_{X|\hat{X}}(x|\hat{x})$.
For $x = 1$: $p = Q^*(1)(1-D) + Q^*(0)D$.
Solving:
\begin{equation}
\label{eq:Qstar}
Q^*(1) = \frac{p - D}{1 - 2D}, \qquad
Q^*(0) = \frac{1 - p - D}{1 - 2D}.
\end{equation}
(Both are positive when $D < \min(p, 1-p)$, which is the interesting regime.)
The forward channel $p^*_{\hat{X}|X}$ is recovered from Bayes' rule,
$p^*_{\hat{X}|X}(\hat{x}|x) = Q^*(\hat{x})\,p^*_{X|\hat{X}}(x|\hat{x})
/p_X(x)$:
\begin{equation}
\label{eq:optimal_test_channel}
p^*_{\hat{X}|X} =
\begin{pmatrix}
\frac{Q^*(0)(1-D)}{1-p} & \frac{Q^*(1)\,D}{1-p} \\[6pt]
\frac{Q^*(0)\,D}{p} & \frac{Q^*(1)(1-D)}{p}
\end{pmatrix}.
\end{equation}
One can verify that each row sums to one.

\begin{remark}
When $p = 1/2$, we have $Q^*(0) = Q^*(1) = 1/2$, and the forward
channel~(\ref{eq:optimal_test_channel}) reduces to a BSC($D$) ---
the symmetric case that many textbooks present as the general answer.
For $p \neq 1/2$, however, the forward channel is \emph{asymmetric}:
the probability of flipping a $0$ to a $1$ differs from the probability of
flipping a $1$ to a $0$.
The key structural insight is that it is the \emph{backward} channel that is
symmetric, not the forward channel.
\end{remark}

The optimal reproduction distribution $Q^*$ was determined
in~(\ref{eq:Qstar}):
\begin{align}
Q^*(0) &= \frac{1 - p - D}{1 - 2D}, \label{eq:repro_dist_0}\\
Q^*(1) &= \frac{p - D}{1 - 2D}. \label{eq:repro_dist_1}
\end{align}
Note that $Q^*(0) + Q^*(1) = 1$, as expected, and $Q^*(1) = p$ only when
$p = 1/2$.
We will use this reproduction distribution again in
Section~\ref{subsec:dtilted} when computing the $d$-tilted information.

The corresponding Lagrange multiplier is
\begin{equation}
\label{eq:lambda_star}
\lambda^* = \log_2 \frac{1-D}{D},
\end{equation}
which is the same value derived in~(\ref{eq:lambda_kkt}).
This is not a coincidence: in any constrained optimization, the Lagrange
multiplier equals the sensitivity of the objective to the constraint.
Here, $\lambda^*$ tells us how much the minimum rate $\RD$ changes when we relax
the distortion budget by a small amount~$dD$.
Formally, $\lambda^* = -R'(D)$, which we can verify directly:
$R'(D) = -H'(D) = \log_2\frac{D}{1-D}$,
so $-R'(D) = \log_2\frac{1-D}{D} = \lambda^*$.
The multiplier is large when $D$ is small (the rate curve is steep, so each
additional bit of distortion saves many bits of rate) and approaches zero as
$D \to \min(p,1-p)$ (the curve flattens near zero rate).

\subsubsection*{Derivation 2: Entropy Maximization}

We now present a second, more direct derivation that bypasses the Lagrangian
machinery entirely.
It relies on three ideas, each of which we explain carefully.

\smallskip
\noindent\textbf{Idea 1: Minimizing mutual information $=$ maximizing
conditional entropy.}
Recall from~(\ref{eq:mi_entropy}) that $I(X;\hat{X}) = H(X) - H(X|\hat{X})$.
Since the source entropy $H(X) = H(p)$ is fixed (it does not depend on the
test channel), minimizing $I(X;\hat{X})$ over the test channel is the same as
maximizing $H(X|\hat{X})$.
Intuitively, $H(X|\hat{X})$ measures how much uncertainty about $X$ remains
\emph{after} seeing $\hat{X}$.
A good lossy code should leave as much residual uncertainty as possible
(retaining only the information needed to meet the distortion target), thereby
minimizing the rate.

\smallskip
\noindent\textbf{Idea 2: Relating conditional entropy to the error variable.}
Define the \emph{error variable} $Z = X \oplus \hat{X}$ (addition modulo~$2$),
so $Z = 1$ when the source and reconstruction disagree and $Z = 0$ otherwise.
Since $X$ is completely determined by the pair $(\hat{X}, Z)$ via $X = \hat{X}
\oplus Z$, knowing $\hat{X}$ and $Z$ tells you $X$ exactly.
By the chain rule of entropy, $H(X|\hat{X}) = H(Z|\hat{X})$:
the residual uncertainty about $X$ given $\hat{X}$ is the same as the
uncertainty about the error pattern $Z$ given $\hat{X}$.

\smallskip
\noindent\textbf{Idea 3: Two upper bounds on $H(Z|\hat{X})$.}
\begin{align}
H(X|\hat{X}) = H(Z|\hat{X})
&\leq H(Z) \label{eq:cond_reduces}\\
&\leq H(D). \label{eq:maxent}
\end{align}
\emph{Why does~(\ref{eq:cond_reduces}) hold?}
This is the general fact that ``conditioning reduces entropy'': knowing
$\hat{X}$ can only help (or not hurt) in predicting $Z$, so
$H(Z|\hat{X}) \leq H(Z)$.
Equality holds if and only if $Z$ is independent of $\hat{X}$ --- that is,
knowing the reconstruction tells you nothing about the error pattern.

\emph{Why does~(\ref{eq:maxent}) hold?}
Since $Z$ is a binary random variable with $\E[Z] = \E[d(X,\hat{X})] \leq D$,
the entropy of $Z$ is at most $H(D)$.
This is because, among all binary random variables with mean at most $D$,
the $\Ber(D)$ distribution has the maximum entropy.
(This is a special case of the general principle that the maximum-entropy
distribution subject to a mean constraint is exponential; for binary variables,
it is Bernoulli.)

\smallskip
\noindent\textbf{Achieving both bounds simultaneously.}
Both inequalities become equalities when:
\begin{enumerate}
    \item $Z$ is independent of $\hat{X}$ (so that $H(Z|\hat{X}) = H(Z)$), and
    \item $Z \sim \Ber(D)$ (so that $H(Z) = H(D)$ and $\E[Z] = D$).
\end{enumerate}
In other words, the optimal scheme has the error $Z \sim \Ber(D)$ acting
independently of the reconstruction $\hat{X}$.
Since $X = \hat{X} \oplus Z$, this means the backward channel
$p_{X|\hat{X}}$ is a BSC with crossover probability $D$, exactly
as~(\ref{eq:backward_channel}).
The reproduction distribution $Q^*$ and forward channel then follow from
Bayes' rule, exactly as in Steps~5--6 of Derivation~1.

% ---- 4.3 ----
\subsection{The Closed-Form Result}
\label{subsec:rd_closedform}

Having identified that the optimal backward channel is BSC($D$), the
rate-distortion function follows immediately.
From the entropy-maximization viewpoint (Derivation~2 of
Section~\ref{subsec:optimal_test_channel}), the maximum conditional entropy is
$H(X|\hat{X}) = H(D)$ (achieved when $Z \sim \Ber(D)$ is independent of
$\hat{X}$), so
\[
\RD = H(p) - \max H(X|\hat{X}) = H(p) - H(D).
\]
Equivalently, from the Lagrangian viewpoint (Derivation~1), $I(X;\hat{X})
= H(X) - H(X|\hat{X}) = H(p) - H(D)$ when evaluated at the optimal test
channel, since the backward BSC($D$) gives $H(X|\hat{X}) = H(D)$.
When $D \leq \min(p, 1-p)$, the distortion constraint is active and the
rate-distortion function is
\begin{equation}
\label{eq:rd_bernoulli}
\boxed{\RD = \Hp - \HD, \qquad 0 \leq D \leq \min(p, 1-p).}
\end{equation}
For $D > \min(p, 1-p)$, we have $\RD = 0$.

\begin{remark}
The formula $\RD = \Hp - \HD$ has an appealing interpretation: the rate equals
the entropy of the source minus the entropy of the ``noise'' introduced by the
test channel.
As the allowed distortion $D$ increases, the noise entropy $\HD$ grows, and
fewer bits are needed to describe the source.
\end{remark}

Let us examine the boundary cases.
When $D = 0$, we require perfect reconstruction, and $R(0) = \Hp$: the rate
equals the source entropy, which is the lossless compression limit.
When $D = \min(p, 1-p)$, the rate drops to zero.
In this regime, the decoder can simply output the more likely symbol ($0$ if
$p < 1/2$, or $1$ if $p > 1/2$) for every position, achieving distortion
$\min(p, 1-p)$ with zero rate.

Figure~\ref{fig:rd_curves} shows the rate-distortion function for several values
of the source bias $p$.

\begin{figure}[htbp]
    \centering
    \includegraphics[width=0.75\textwidth]{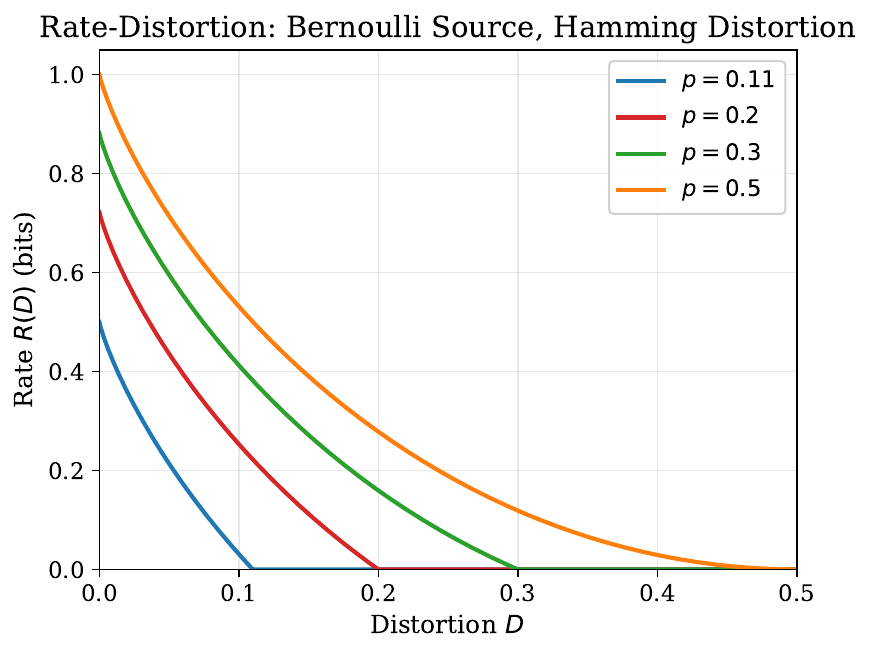}
    \caption{The rate-distortion function $\RD = \Hp - \HD$ for a $\Ber(p)$
    source with Hamming distortion, shown for $p \in \{0.11, 0.2, 0.3, 0.5\}$.
    Each curve is convex and decreasing, starting at $R(0) = \Hp$ and reaching
    zero at $D = \min(p, 1-p)$.
    The $p = 0.5$ curve starts highest because the fair coin has the most entropy.}
    \label{fig:rd_curves}
\end{figure}

We note three important properties of the rate-distortion function:
\begin{enumerate}
    \item \textbf{Convexity:} $\RD$ is a convex function of $D$.
          This means that each additional unit of distortion ``buys'' progressively
          less rate reduction.
    \item \textbf{Monotonicity:} $\RD$ is non-increasing in $D$.
          Allowing more distortion can only help (or leave unchanged) the compression rate.
    \item \textbf{Continuity:} $\RD$ is continuous on $[0, \min(p, 1-p)]$.
\end{enumerate}

% ---- 4.4 ----
\subsection{Historical Note}
\label{subsec:rd_history}

Shannon stated the rate-distortion function for the binary source in his 1959
paper~\cite{shannon1959}.
A comprehensive treatment of rate-distortion theory for general sources was
developed by Berger~\cite{berger1971}.
The elegant formula $\RD = \Hp - \HD$ serves as the starting point for
virtually every textbook discussion of lossy source coding; see, for
example, Cover and Thomas~\cite{cover2006}.

% ========================================================================
% SECTION 5: THE BLAHUT-ARIMOTO ALGORITHM
% ========================================================================
\section{The Blahut-Arimoto Algorithm}
\label{sec:blahut_arimoto}

In this section, we present the Blahut-Arimoto algorithm, a powerful iterative
method for computing rate-distortion functions.
While the Bernoulli source admits a closed-form solution, the Blahut-Arimoto
algorithm applies to arbitrary finite-alphabet sources and distortion measures.

% ---- 5.1 ----
\subsection{Motivation}
\label{subsec:ba_motivation}

The rate-distortion function~(\ref{eq:rd_function}) is defined as a minimization
of mutual information over a convex set of conditional distributions.
For the Bernoulli source with Hamming distortion, we exploited the problem's
symmetry to obtain a closed-form solution.
However, for more complex sources, such as non-uniform discrete sources with
non-binary alphabets or non-Hamming distortion measures, the optimization
does not admit a closed-form solution, and a computational approach is needed.

The Blahut-Arimoto algorithm~\cite{blahut1972,arimoto1972} solves this
optimization through an elegant \emph{alternating minimization} procedure.
It was independently discovered by Blahut and Arimoto in 1972, and its
convergence was later established rigorously by Csisz\'{a}r~\cite{csiszar1974}.

% ---- 5.2 ----
\subsection{The Algorithm}
\label{subsec:ba_algorithm}

The Blahut-Arimoto algorithm operates on the Lagrangian dual formulation of
the rate-distortion problem.
For a given slope parameter $s > 0$ (corresponding to the Lagrange multiplier),
the algorithm minimizes the functional
\begin{equation}
\label{eq:ba_functional}
F(s) = \min_{p_{\hat{X}|X}} \bigl[ I(X; \hat{X}) + s \cdot \E[d(X, \hat{X})] \bigr].
\end{equation}
By sweeping $s$ over positive values, we trace out the entire $\RD$ curve.
(The slope $s$ here is in nats and relates to the Lagrange multiplier $\lambda^*$
from Section~\ref{subsec:optimal_test_channel} by $s = \lambda^* \ln 2$; the natural
exponential $e^{-s\,d}$ in Step~1 below is equivalent to the
$2^{-\lambda\,d}$ form used there.)

The algorithm alternates between two updates:

\smallskip
\noindent\textbf{Step 1: Update the test channel.}
Given the current reproduction distribution $p_{\hat{X}}(\hat{x})$, update
\begin{equation}
\label{eq:ba_step1}
p_{\hat{X}|X}(\hat{x}|x) = \frac{p_{\hat{X}}(\hat{x})\, e^{-s\, d(x,\hat{x})}}{Z(x)},
\end{equation}
where $Z(x) = \sum_{\hat{x}} p_{\hat{X}}(\hat{x})\, e^{-s\, d(x,\hat{x})}$ is
a normalization constant ensuring $\sum_{\hat{x}} p_{\hat{X}|X}(\hat{x}|x) = 1$.

\smallskip
\noindent\textbf{Step 2: Update the reproduction distribution.}
Given the updated test channel, compute
\begin{equation}
\label{eq:ba_step2}
p_{\hat{X}}(\hat{x}) = \sum_{x} p_X(x)\, p_{\hat{X}|X}(\hat{x}|x).
\end{equation}
This is simply the marginal of $\hat{X}$ induced by passing $X$ through the
updated test channel.

The algorithm is initialized with a uniform reproduction distribution
$p_{\hat{X}}(\hat{x}) = 1/|\calXhat|$ and iterates Steps~1 and~2 until
convergence.

\begin{remark}
The alternating structure of the algorithm has a natural interpretation:
Step~1 finds the best test channel for a fixed output distribution, while
Step~2 updates the output distribution to be consistent with the new test channel.
This is an instance of the classical alternating minimization framework, and
convergence is guaranteed because each step decreases the Lagrangian objective.
\end{remark}

\begin{remark}[Intuition for Step~1]
The update in~(\ref{eq:ba_step1}) has an appealing interpretation: for each source
symbol $x$, the algorithm reweights the current reproduction distribution
$p_{\hat{X}}(\hat{x})$ by the factor $e^{-s\,d(x,\hat{x})}$.
Reproduction symbols $\hat{x}$ that are close to $x$ (low distortion) receive a
weight near~$1$, while those that are far from $x$ (high distortion) are
exponentially suppressed.
The result is then normalized to form a valid conditional distribution.
This is precisely a \emph{multiplicative weights update}: rather than making
additive adjustments, the algorithm multiplies each probability by an
exponential penalty that depends on the ``cost'' $d(x,\hat{x})$.
Readers familiar with online learning will recognize the same mechanism at work
in the Hedge algorithm and its bandit variant Exp3, where actions are reweighted
by $e^{-\eta \cdot \text{loss}}$ to balance exploration and exploitation.
In the Blahut-Arimoto setting, the slope parameter $s$ plays the role of the
learning rate $\eta$: a larger $s$ imposes a harsher penalty on distortion,
concentrating the test channel on the closest reproduction symbols.
\end{remark}

Algorithm~\ref{alg:ba} presents the complete pseudocode.

\begin{algorithm}[t]
\caption{Blahut-Arimoto Algorithm}
\label{alg:ba}
\begin{algorithmic}[1]
\Require Source distribution $p_X$, distortion matrix $d$, slope $s > 0$, tolerance $\delta$
\Ensure Rate $R$ and distortion $D$ on the $\RD$ curve
\State Initialize $p_{\hat{X}}(\hat{x}) \gets 1/|\calXhat|$ for all $\hat{x}$
\Repeat
    \For{each $x \in \calX$}
        \State $Z(x) \gets \sum_{\hat{x}} p_{\hat{X}}(\hat{x})\, e^{-s\, d(x,\hat{x})}$
        \For{each $\hat{x} \in \calXhat$}
            \State $p_{\hat{X}|X}(\hat{x}|x) \gets p_{\hat{X}}(\hat{x})\, e^{-s\, d(x,\hat{x})} / Z(x)$
        \EndFor
    \EndFor
    \For{each $\hat{x} \in \calXhat$}
        \State $p_{\hat{X}}(\hat{x}) \gets \sum_{x} p_X(x)\, p_{\hat{X}|X}(\hat{x}|x)$
    \EndFor
    \State Compute $R \gets I(X; \hat{X})$ and $D \gets \E[d(X,\hat{X})]$
\Until{$|R_{\text{new}} - R_{\text{old}}| < \delta$}
\State \Return $(R, D)$
\end{algorithmic}
\end{algorithm}

% ---- 5.3 ----
\subsection{Application to the Bernoulli Source}
\label{subsec:ba_bernoulli}

We now apply the Blahut-Arimoto algorithm to the $\Ber(p)$ source with Hamming
distortion.
Since the source and reproduction alphabets are both $\{0, 1\}$, the test
channel is a $2 \times 2$ matrix, and all computations are explicit.

Consider a concrete example with $p = 0.3$ and slope parameter $s = 10$.
At each iteration, we maintain the reproduction distribution
$(p_{\hat{X}}(0), p_{\hat{X}}(1))$ and the test channel matrix.
The updates in~(\ref{eq:ba_step1}) and~(\ref{eq:ba_step2}) become:
\begin{align}
p_{\hat{X}|X}(0|0) &= \frac{p_{\hat{X}}(0)}{p_{\hat{X}}(0) + p_{\hat{X}}(1)\, e^{-s}}, \\
p_{\hat{X}|X}(0|1) &= \frac{p_{\hat{X}}(0)\, e^{-s}}{p_{\hat{X}}(0)\, e^{-s} + p_{\hat{X}}(1)}.
\end{align}
Since $e^{-s} = e^{-10} \approx 4.5 \times 10^{-5}$, the multiplicative weight
for any mismatched reconstruction ($d(x,\hat{x}) = 1$) is tiny compared to the
weight for a correct reconstruction ($d(x,\hat{x}) = 0$, weight~$1$).
For example, in the numerator of $p_{\hat{X}|X}(0|0)$, the term $p_{\hat{X}}(0)$
(correct match) appears with weight~$1$, while the denominator also includes
$p_{\hat{X}}(1)\,e^{-10}$ (mismatch), which is roughly $10^5$ times smaller.
The result is that $p_{\hat{X}|X}(0|0) \approx 1$ and $p_{\hat{X}|X}(1|0) \approx 0$:
the test channel is driven toward a near-identity matrix, corresponding to very
low distortion.

Figure~\ref{fig:ba_convergence} shows the convergence of the Blahut-Arimoto
algorithm for $p = 0.3$ and several values of the slope parameter $s$.
The algorithm converges rapidly, typically reaching machine precision within
$20$--$50$ iterations.
Larger values of $s$ correspond to lower target distortions and converge faster.

\begin{figure}[htbp]
    \centering
    \includegraphics[width=0.75\textwidth]{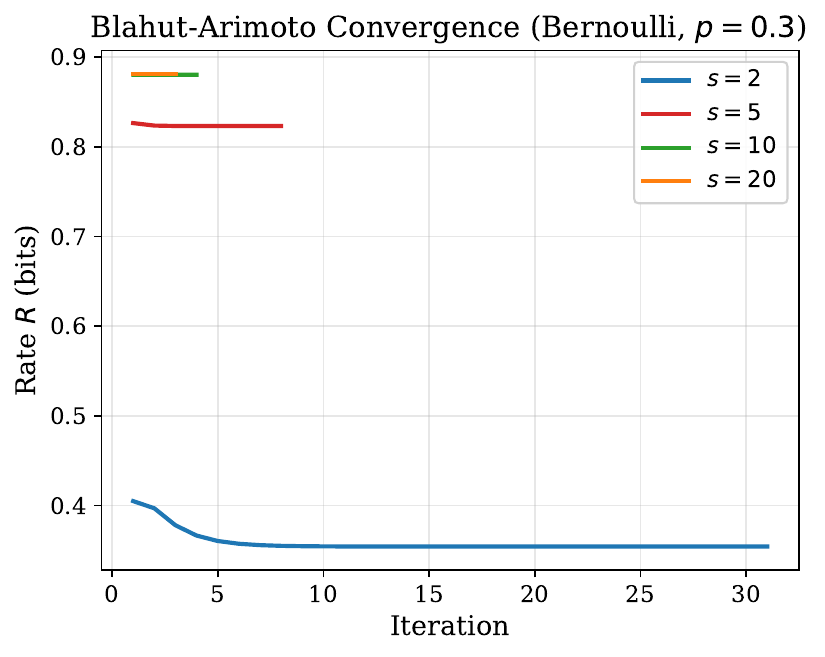}
    \caption{Convergence of the Blahut-Arimoto algorithm for $p = 0.3$ and slope
    parameters $s \in \{2, 5, 10, 20\}$.
    The rate converges monotonically to its final value within a few tens of iterations.}
    \label{fig:ba_convergence}
\end{figure}

To validate the algorithm, we sweep $s$ over a range of values and plot the
resulting $(D, R)$ pairs alongside the closed-form curve $\RD = \Hp - \HD$.
Figure~\ref{fig:ba_vs_closedform} confirms that the Blahut-Arimoto algorithm
recovers the exact rate-distortion function.

\begin{figure}[htbp]
    \centering
    \includegraphics[width=0.75\textwidth]{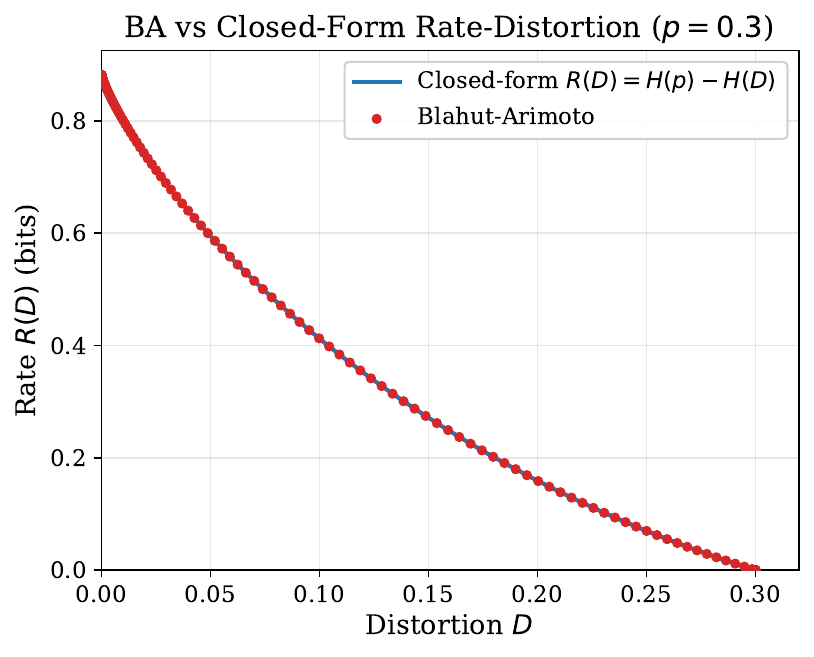}
    \caption{Comparison of the Blahut-Arimoto computed rate-distortion points
    (circles) with the closed-form curve $\RD = \Hp - \HD$ (solid line)
    for $p = 0.3$.
    The agreement is exact to numerical precision.}
    \label{fig:ba_vs_closedform}
\end{figure}

% ---- 5.4 ----
\subsection{A Concrete Worked Example}
\label{subsec:worked_example}

To bridge the gap between the abstract Blahut-Arimoto algorithm and the finite
block length theory of Section~\ref{sec:finite_blocklength}, we construct
a concrete $(n, M)$ code for the $\Ber(0.3)$ source.  With $n = 3$ and $M = 4$
codewords, the optimal codebook is $\mathcal{C} = \{000, 001, 010, 100\}$: all
sequences of Hamming weight~$0$ or~$1$.  These are the four most probable source
sequences, and a nearest-neighbor encoder maps each remaining sequence to the
closest codeword.  Figure~\ref{fig:codebook_example} displays the complete code.

\begin{figure}[htbp]
    \centering
    \includegraphics[width=\textwidth]{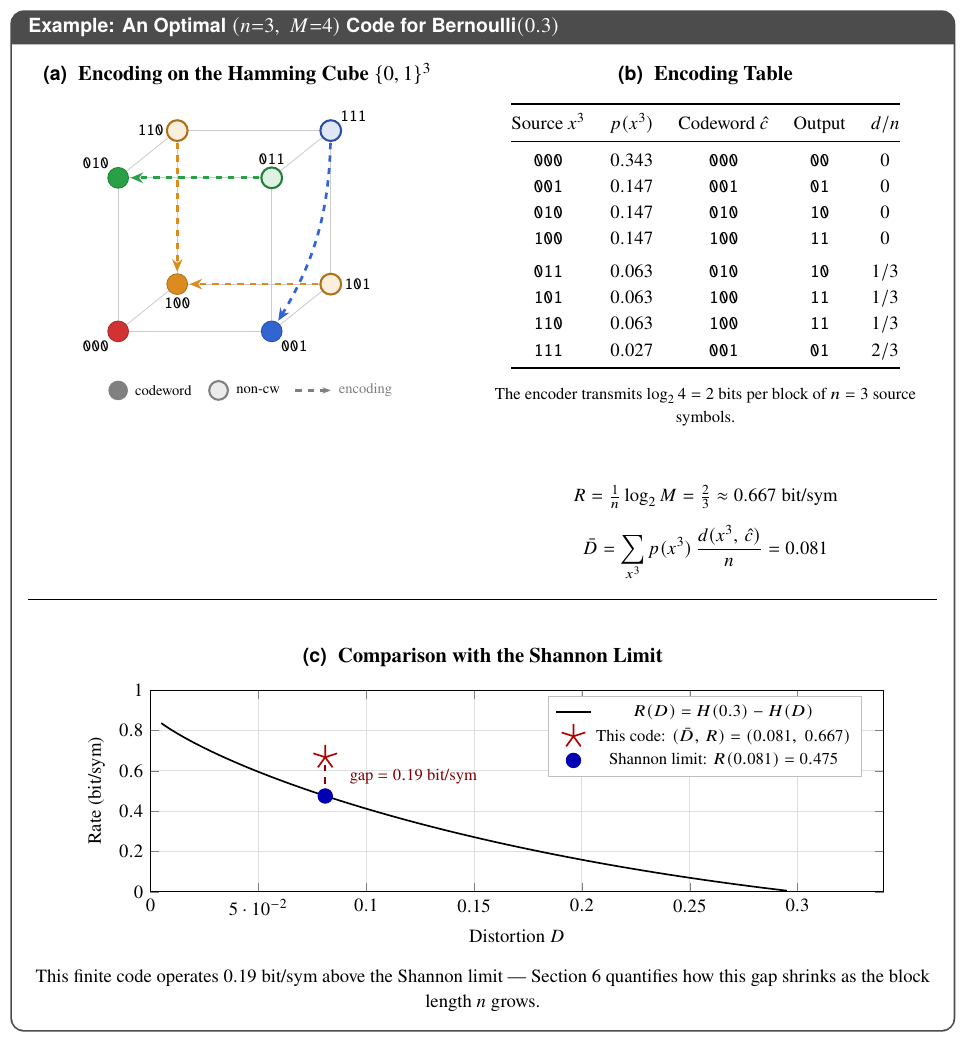}
    \caption{A worked example of an optimal $(n{=}3,\; M{=}4)$ lossy code for
    $\Ber(0.3)$ with Hamming distortion.
    \textsf{(a)}~The Hamming cube $\{0,1\}^3$ with codewords shown as filled
    circles and non-codewords as hollow circles; dashed arrows indicate the
    nearest-neighbor encoding map.
    \textsf{(b)}~The encoding table listing every source sequence, its
    probability, assigned codeword, and per-symbol distortion.
    \textsf{(c)}~The code's operating point $(\bar{D}, R) = (0.081, 0.667)$
    compared with the Shannon limit $R(0.081) = 0.475$.}
    \label{fig:codebook_example}
\end{figure}

The code achieves rate $R = \frac{2}{3} \approx 0.667$ bit/sym at average
distortion $\bar{D} = 0.081$, which is $0.19$ bit/sym above the Shannon limit
$R(0.081) = 0.475$.  This gap is the price of operating at finite block
length --- Section~\ref{sec:finite_blocklength} quantifies precisely how it
shrinks as $n$ grows.

% ---- 5.5 ----
\subsection{Historical Note}
\label{subsec:ba_history}

The algorithm was independently proposed by Blahut~\cite{blahut1972} and
Arimoto~\cite{arimoto1972} in 1972.
Blahut's formulation emphasized the Lagrangian dual structure, while Arimoto
approached the problem through an iterative projection method.
Csisz\'{a}r~\cite{csiszar1974} unified and extended both approaches using
his theory of $I$-divergence, establishing the alternating minimization
interpretation that clarifies why the iterations converge.
The Blahut-Arimoto algorithm remains the standard computational tool for
rate-distortion functions and channel capacities in information theory.

% ========================================================================
% SECTION 6: BEYOND THE ASYMPTOTIC LIMIT
% ========================================================================
\section{Beyond the Asymptotic Limit: Finite Block Length}
\label{sec:finite_blocklength}

In this section, we move beyond Shannon's asymptotic rate-distortion function
and develop the theory of finite block length lossy compression.
This is the mathematical core of the tutorial.

% ---- 6.1 ----
\subsection{The Gap Between Theory and Practice}
\label{subsec:gap}

The rate-distortion function $\RD$ tells us the ultimate limit of lossy
compression as the block length $n \to \infty$.
However, real systems operate with finite $n$.
A practical compression system might use blocks of $n = 100$ or $n = 1000$
symbols.
How much extra rate do we need compared to the Shannon limit?

Figure~\ref{fig:rate_vs_blocklength} illustrates the situation.
For a $\Ber(0.3)$ source with target distortion $D = 0.1$ and excess-distortion
probability $\varepsilon = 0.1$, the achievable rate at $n = 100$ is
significantly above $\RD$, but the gap narrows as $n$ grows.
Understanding the precise rate of this convergence is the goal of finite block
length theory.

\begin{figure}[htbp]
    \centering
    \includegraphics[width=0.75\textwidth]{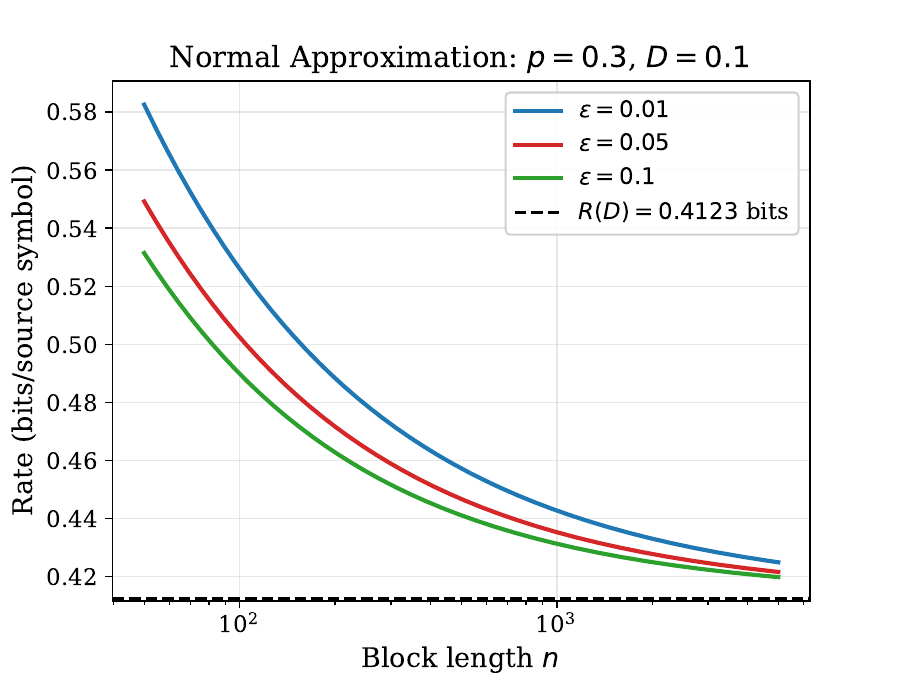}
    \caption{The minimum achievable rate $R(n, D, \varepsilon)$ versus block length
    $n$ for a $\Ber(0.3)$ source with $D = 0.1$ and several values of the
    excess-distortion probability $\varepsilon$.
    The horizontal dashed line shows the Shannon limit $\RD$.
    The gap decays as $O(1/\sqrt{n})$.}
    \label{fig:rate_vs_blocklength}
\end{figure}

% ---- 6.2 ----
\subsection{The Finite Block Length Setup}
\label{subsec:fbl_setup}

Before diving into the formal definitions, let us build some intuition for why
the finite block length setting requires a fundamentally different formulation
than the asymptotic one.

The asymptotic rate-distortion function $\RD$ is a deterministic quantity.
This is because, in the limit $n \to \infty$, the law of large numbers
guarantees that the empirical statistics of the source sequence $X^n$ concentrate
around their expected values.
Almost every source sequence ``looks the same'' statistically, so a single
codebook can handle all of them with average distortion at most $D$.
There is no randomness left to worry about.

At finite block length, the situation is different.
The source sequence $X^n$ is random, and different realizations can be
substantially easier or harder to compress.
For a $\Ber(0.3)$ source with $n = 100$, most sequences will have roughly
$30$ ones, but some will have $15$ or $45$.
A code designed for the ``typical'' case may fail badly on these atypical
sequences, producing a reconstruction whose distortion exceeds $D$.
In short, the distortion achieved by any finite-$n$ code is itself a random
variable, because it depends on which source sequence nature produces.

This means we cannot simply ask for the distortion to be at most $D$ with
certainty (that would require an enormous codebook) or merely on average (that
would hide the possibility of catastrophic failures on some sequences).
Instead, we allow a small \emph{failure probability} $\varepsilon$: we accept
that a fraction~$\varepsilon$ of source sequences may result in distortion
exceeding~$D$, but we require the code to succeed on the remaining
$1 - \varepsilon$ fraction.
The quantity $R(n, D, \varepsilon)$ is then deterministic again: it is the
minimum rate at which a code of block length $n$ exists that keeps the
failure probability below~$\varepsilon$.
One can think of $R(n, D, \varepsilon)$ as a \emph{confidence bound}: ``with
confidence $1 - \varepsilon$, a rate of $R(n, D, \varepsilon)$ bits per symbol
suffices to compress the source to distortion at most~$D$.''

We now formalize this precisely.

\begin{definition}[$(n, M, D, \varepsilon)$ Code]
An $(n, M, D, \varepsilon)$ lossy source code consists of an encoder
$f_n : \calX^n \to \{1, \ldots, M\}$ and a decoder
$g_n : \{1, \ldots, M\} \to \calXhat^n$ such that the
\emph{excess-distortion probability} satisfies
\begin{equation}
\label{eq:excess_distortion}
\Prob\bigl(d(X^n, g_n(f_n(X^n))) > D\bigr) \leq \varepsilon.
\end{equation}
\end{definition}

Note the shift from the asymptotic setting.
Instead of requiring the \emph{average} distortion to be at most $D$, we
require that the distortion exceeds $D$ with probability at most $\varepsilon$.
This excess-distortion formulation is more natural for finite block lengths and
leads to cleaner second-order results.

The \emph{minimum achievable rate} at block length $n$ is
\begin{equation}
\label{eq:Rn}
R(n, D, \varepsilon) = \frac{1}{n} \log_2 M^*(n, D, \varepsilon),
\end{equation}
where $M^*(n, D, \varepsilon)$ is the smallest codebook size $M$ for which an
$(n, M, D, \varepsilon)$ code exists.
The fundamental result of finite block length theory is a precise
characterization of $R(n, D, \varepsilon)$.

% ---- 6.3 ----
\subsection{A Type-Based Heuristic Route}
\label{subsec:type_derivation}

Before giving the rigorous $d$-tilted-information analysis, we present a
Bernoulli-specific heuristic based on types that correctly predicts the normal
approximation and dispersion~\cite{ingber2011}.
This approach uses the method of types and the central limit theorem;
it is not a proof, but it gives a transparent and intuitive account of why the
dispersion takes the form it does.

\paragraph{The type summarizes the block.}
A source sequence $x^n \in \{0,1\}^n$ drawn i.i.d.\ from $\Ber(p)$ has
$k = \sum_{i=1}^{n} x_i$ ones.
The integer $k$ determines the \emph{type} (empirical distribution)
$\hat{p}_n = k/n$, and the \emph{type class}
$T_k = \{x^n \in \{0,1\}^n : \sum_i x_i = k\}$ has size $\binom{n}{k}$.
The sufficient statistic for compression is just the count
$K = \sum_i X_i \sim \mathrm{Bin}(n, p)$.

\paragraph{Rate needed for a given type.}
For a sequence of type $k/n$, the rate-distortion function applied to the
empirical distribution gives a ``local'' rate
\begin{equation}
\label{eq:local_rate}
R(k/n,\, D) \;=\; H(k/n) - \HD,
\end{equation}
provided $k/n \geq D$ (otherwise $R = 0$).
At the exponential scale, covering all $\binom{n}{k} \approx 2^{nH(k/n)}$
sequences in $T_k$ by Hamming balls of radius $nD$ (each of size roughly
$2^{n\HD}$) requires about $2^{n[H(k/n) - \HD]}$ codewords.

\paragraph{Excess-distortion as a type fluctuation.}
If we use a codebook of rate $R$, it suffices for type class $T_k$ whenever
$H(k/n) - \HD \leq R$.
A code fails when the source sequence has ``too many'' or ``too few'' ones---that
is, when the empirical entropy $H(\hat{p}_n)$ is too large.
The excess-distortion event therefore reduces to a fluctuation of the type:
\[
\Prob\!\bigl(d(X^n, g_n(f_n(X^n))) > D\bigr)
\;\approx\;
\Prob\!\bigl(H(\hat{p}_n) - \HD > R\bigr).
\]

\paragraph{Exact binomial formulation.}
Define $g_n(k) = H(k/n) - \HD$ for $k = 0, 1, \ldots, n$.
The excess-distortion probability at rate $R$ is then
\begin{equation}
\label{eq:type_excess}
\Prob\!\bigl(g_n(K) > R\bigr)
\;=\; \sum_{\substack{k=0 \\[1pt] g_n(k) > R}}^{n}
\binom{n}{k}\, p^k (1-p)^{n-k}.
\end{equation}
The \emph{type-based optimal rate} is
\[
R_{\mathrm{type}}(n, D, \varepsilon)
\;=\; \inf\!\Bigl\{R \geq 0 :
\Prob\!\bigl(g_n(K) > R\bigr) \leq \varepsilon\Bigr\}.
\]
This is the exact binomial-quantile version of the finite block length
rate-distortion problem for the Bernoulli source.

\paragraph{Gaussian approximation via the delta method.}
By the central limit theorem,
$\hat{p}_n = K/n \approx \mathcal{N}\!\bigl(p,\; p(1-p)/n\bigr)$
for large $n$.
The function $g(\hat{p}) = H(\hat{p}) - \HD$ is smooth in $\hat{p}$, so the
delta method gives
\[
g(\hat{p}_n) \;\approx\;
\mathcal{N}\!\Bigl(g(p),\;
\bigl[g'(p)\bigr]^2 \,\frac{p(1-p)}{n}\Bigr).
\]
The derivative of the binary entropy is
$g'(\hat{p}) = H'(\hat{p}) = \log_2\!\frac{1-\hat{p}}{\hat{p}}$,
so $g'(p) = \log_2\!\frac{1-p}{p}$ and
\[
\Var\!\bigl(g(\hat{p}_n)\bigr)
\;\approx\;
\frac{1}{n}\,p(1-p)\,\Bigl[\log_2\frac{1-p}{p}\Bigr]^{\!2}
\;=\; \frac{V}{n},
\]
where
\begin{equation}
\label{eq:V_type}
V \;=\; p(1-p)\,\Bigl[\log_2\frac{1-p}{p}\Bigr]^{\!2}.
\end{equation}
Setting
$\Prob\!\bigl(g(\hat{p}_n) > R\bigr) = \varepsilon$ and inverting the
Gaussian tail gives the normal approximation
\[
R(n, D, \varepsilon)
\;\approx\;
\RD + \sqrt{\frac{V}{n}}\;\Qinv(\varepsilon)
+ o\!\left(\frac{1}{\sqrt{n}}\right).
\]

\paragraph{$V$ does not depend on $D$.}
Equation~(\ref{eq:V_type}) is striking: the dispersion $V$ depends on the
source parameter $p$ but \emph{not} on the target distortion $D$.
The reason is immediate from the type viewpoint: the randomness in
$g(\hat{p}_n) = H(\hat{p}_n) - \HD$ comes entirely from $H(\hat{p}_n)$;
the term $-\HD$ is a deterministic shift that affects the mean but not the
variance.
This $D$-independence is special to the Bernoulli source with Hamming
distortion and does not hold for general sources.

\paragraph{Caveat.}
The type-based argument above treats the covering number
$|T_k|/|B(x, nD)|$ as exact;
in reality, this ratio is only a first-order (exponential-scale) approximation
to the true codebook size.
Consequently, this derivation is a heuristic sketch, not a rigorous
proof---the rigorous bound requires the $d$-tilted information framework
developed in the next subsection.

\paragraph{Transition.}
The type-based route works beautifully for the Bernoulli source because the
empirical entropy $H(\hat{p}_n)$ captures all the relevant randomness.
For general sources and distortion measures, there is no equally simple
sufficient statistic, and the type argument breaks down.
The \emph{$d$-tilted information}, introduced next, provides the
generalization: it reduces to $H(\hat{p}_n) - \HD$ for the Bernoulli case but
applies to arbitrary memoryless sources.
For Bernoulli/Hamming, however, the type-based fluctuation picture is not
merely motivational: the $d$-tilted information is exactly
$\jmath_X(x, D) = H(\hat{p}_n) - \HD$ evaluated at the source realization, so
the heuristic argument closely mirrors the structure behind the rigorous proof.

% ---- 6.4 ----
\subsection{The $d$-Tilted Information}
\label{subsec:dtilted}

The central single-letter quantity in the finite block length analysis is the
\emph{$d$-tilted information}, introduced by Kostina and
Verd\'{u}~\cite{kostina2012}.

\begin{definition}[$d$-Tilted Information {\cite{kostina2012}}]
\label{def:dtilted}
For a source $X$ with distribution $p_X$, distortion measure $d$, and target
distortion $D$, the $d$-tilted information of a source realization $x$ is
\begin{equation}
\label{eq:dtilted}
\jmath_X(x, D) = D_{\mathrm{KL}}\!\bigl(p^*_{\hat{X}|X}(\cdot|x) \,\big\|\, Q^*\bigr)
+ \lambda^*\!\bigl(\E[d(x,\hat{X}) \mid X = x] - D\bigr),
\end{equation}
where $\lambda^* = \log_2\frac{1-D}{D}$ is the optimal Lagrange multiplier~(\ref{eq:lambda_star}),
$Q^*$ is the optimal reproduction distribution~(\ref{eq:repro_dist_0})--(\ref{eq:repro_dist_1}),
and $p^*_{\hat{X}|X}$ is the optimal forward channel~(\ref{eq:optimal_test_channel}).
All quantities are measured in bits (using $\log_2$).
Equivalently, using the Gibbs form $p^*_{\hat{X}|X}(\hat{x}|x)
= Q^*(\hat{x})\,2^{-\lambda^* d(x,\hat{x})}/Z(x)$ with normalizing constant
$Z(x) = \sum_{\hat{x}} Q^*(\hat{x})\,2^{-\lambda^* d(x,\hat{x})}$,
the definition simplifies to
\begin{equation}
\label{eq:dtilted_bits}
\jmath_X(x, D) = -D\log_2\frac{1-D}{D}
+ \log_2 \frac{1}{Z(x)}.
\end{equation}
\end{definition}

The $d$-tilted information has a compelling interpretation: it measures how
``difficult'' it is to compress a particular source realization $x$ to distortion
level $D$.
Different source symbols may be easier or harder to compress, and $\jmath_X(x, D)$
captures this variation.

A key property is that the expected $d$-tilted information equals the
rate-distortion function:
\begin{equation}
\label{eq:dtilted_mean}
\E[\jmath_X(X, D)] = \RD,
\end{equation}
where the expectation is over the source distribution:
$\E[\jmath_X(X, D)] = \sum_{x} p_X(x)\,\jmath_X(x, D)
= (1-p)\,\jmath_X(0, D) + p\,\jmath_X(1, D)$ for the Bernoulli source.

\begin{proof}
We give two proofs: an explicit algebraic computation for the Bernoulli source,
and a short information-theoretic argument that works in general.

\smallskip
\noindent\textbf{Part A: Algebraic proof for the Bernoulli source.}
We use the equivalent form~(\ref{eq:dtilted_bits}).
The normalizing constant is $Z(x) = \sum_{\hat{x}} Q^*(\hat{x})\,
2^{-\lambda^* d(x,\hat{x})}$ with $Q^*$ from~(\ref{eq:Qstar})
and $\lambda^* = \log_2\frac{1-D}{D}$, so $2^{-\lambda^*} = \frac{D}{1-D}$.

\smallskip
\noindent\emph{Step 1: Compute $Z(0)$ and $Z(1)$.}
\begin{align}
Z(0) &= Q^*(0)\cdot 1 + Q^*(1)\cdot\frac{D}{1-D}
= \frac{Q^*(0)(1-D) + Q^*(1)\,D}{1-D}. \notag
\end{align}
Substituting $Q^*(0) = \frac{1-p-D}{1-2D}$ and $Q^*(1) = \frac{p-D}{1-2D}$:
\[
Q^*(0)(1-D) + Q^*(1)\,D
= \frac{(1-p-D)(1-D) + (p-D)D}{1-2D}.
\]
Expanding the numerator:
$(1-p-D)(1-D) + (p-D)D = (1-p)(1-D) - D(1-D) + pD - D^2
= (1-p)(1-D) - D(1-p) = (1-p)(1-2D)$.
Therefore
\begin{equation}
\label{eq:Z0Z1}
Z(0) = \frac{1-p}{1-D}, \qquad
Z(1) = \frac{p}{1-D},
\end{equation}
where $Z(1)$ follows by an identical calculation (or by the symmetry
$Z(1) = Q^*(0)\frac{D}{1-D} + Q^*(1)$, which gives numerator $p(1-2D)$).

\smallskip
\noindent\emph{Step 2: Compute $\jmath_X(0,D)$ and $\jmath_X(1,D)$.}
Plugging~(\ref{eq:Z0Z1}) into~(\ref{eq:dtilted_bits}):
\begin{align}
\jmath_X(0, D) &= -D\log_2\frac{1-D}{D} + \log_2\frac{1-D}{1-p},
\label{eq:jx0_simplified}\\
\jmath_X(1, D) &= -D\log_2\frac{1-D}{D} + \log_2\frac{1-D}{p}.
\label{eq:jx1_simplified}
\end{align}

\smallskip
\noindent\emph{Step 3: Compute the expectation.}
\begin{align}
\E[\jmath_X(X,D)]
&= (1-p)\,\jmath_X(0,D) + p\,\jmath_X(1,D) \notag\\
&= \underbrace{-D\log_2\frac{1-D}{D}}_{\text{common first term}}
   + \underbrace{(1-p)\log_2\frac{1-D}{1-p} + p\log_2\frac{1-D}{p}}_{\text{weighted second terms}}.
   \label{eq:Ej_expanded}
\end{align}

\smallskip
\noindent\emph{Step 4: Simplify the second group.}
Splitting the logarithms:
\[
(1-p)\log_2\frac{1-D}{1-p} + p\log_2\frac{1-D}{p}
= \underbrace{\bigl[(1-p) + p\bigr]}_{=\,1}\log_2(1-D)
  \;+\; \underbrace{\bigl[-(1-p)\log_2(1-p) - p\log_2 p\bigr]}_{=\,H(p)}.
\]

\smallskip
\noindent\emph{Step 5: Combine all terms.}
\begin{align}
\E[\jmath_X(X,D)]
&= -D\log_2\frac{1-D}{D} + \log_2(1-D) + H(p) \notag\\
&= -D\log_2(1-D) + D\log_2 D + \log_2(1-D) + H(p) \notag\\
&= (1-D)\log_2(1-D) + D\log_2 D + H(p) \notag\\
&= -H(D) + H(p) = H(p) - H(D) = \RD. \qquad\square \notag
\end{align}

\smallskip
\noindent\textbf{Part B: Information-theoretic proof (general sources).}
From the definition~(\ref{eq:dtilted}), taking the expectation over $X$:
\begin{align}
\E[\jmath_X(X,D)]
&= \E_X\!\bigl[D_{\mathrm{KL}}(p^*_{\hat{X}|X}(\cdot|X) \| Q^*)\bigr]
   + \lambda^*\!\bigl(\E[d(X,\hat{X}^*)] - D\bigr). \notag
\end{align}
The first term is $\sum_x p_X(x)\,D_{\mathrm{KL}}(p^*_{\hat{X}|X}(\cdot|x) \| Q^*)
= I(X;\hat{X}^*)$, since mutual information decomposes as the expected
KL divergence of the conditional from the marginal.
The second term vanishes because the optimal test channel meets the distortion
constraint with equality: $\E[d(X,\hat{X}^*)] = D$.
Therefore
\[
\E[\jmath_X(X,D)] = I(X;\hat{X}^*) = \RD. \qquad\square
\]

For the Bernoulli source, the accompanying Python code confirms that
$(1-p)\,\jmath_X(0,D) + p\,\jmath_X(1,D) = \Hp - \HD$
to machine precision for all tested values of $p$ and $D$.
\end{proof}

This identity confirms that $\jmath_X(x, D)$ is the correct ``information density''
for the lossy compression problem: it decomposes the rate $\RD$ into per-symbol
contributions, just as the log-likelihood ratio decomposes mutual information
in channel coding.

For the $\Ber(p)$ source with Hamming distortion, the optimal reproduction
distribution is given
by~(\ref{eq:repro_dist_0})--(\ref{eq:repro_dist_1}), and the Lagrange
multiplier is $\lambda^* = \log_2\frac{1-D}{D}$.
Using the simplified normalizing constants $Z(0) = \frac{1-p}{1-D}$
and $Z(1) = \frac{p}{1-D}$ from~(\ref{eq:Z0Z1}), the $d$-tilted information
takes two values:
\begin{align}
\jmath_X(0, D) &= -D\log_2\frac{1-D}{D} + \log_2\frac{1-D}{1-p}, \label{eq:jx0}\\
\jmath_X(1, D) &= -D\log_2\frac{1-D}{D} + \log_2\frac{1-D}{p}. \label{eq:jx1}
\end{align}

Figure~\ref{fig:dtilted} shows the $d$-tilted information for both source
symbols as a function of $D$.

\begin{figure}[htbp]
    \centering
    \includegraphics[width=0.75\textwidth]{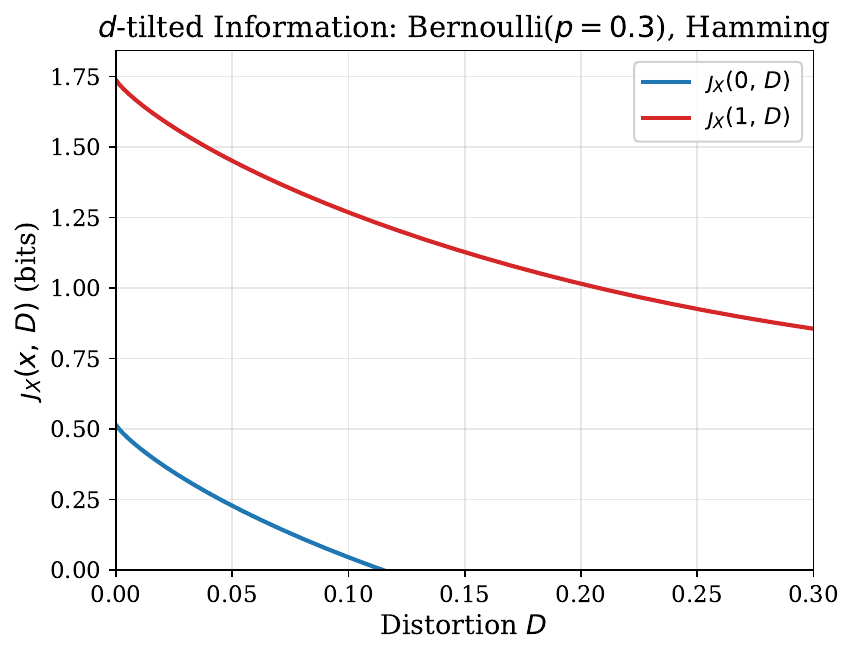}
    \caption{The $d$-tilted information $\jmath_X(0, D)$ and $\jmath_X(1, D)$
    for a $\Ber(0.3)$ source with Hamming distortion.
    When $D$ is small, both values are close to $\Hp$ (the lossless rate).
    As $D$ increases toward $\min(p, 1-p) = 0.3$, both converge to zero.
    The gap between the two curves reflects the asymmetry of the source.}
    \label{fig:dtilted}
\end{figure}

% ---- 6.5 ----
\subsection{Dispersion: The Key Second-Order Quantity}
\label{subsec:dispersion}

The rate-distortion function $\RD$ is a first-order quantity: it captures the
leading-order behavior as $n \to \infty$.
The \emph{rate-distortion dispersion} is the second-order quantity that governs
how quickly the finite block length rate converges to $\RD$.

\begin{definition}[Rate-Distortion Dispersion]
\label{def:dispersion}
The rate-distortion dispersion at distortion level $D$ is
\begin{equation}
\label{eq:dispersion}
V(D) = \Var[\jmath_X(X, D)].
\end{equation}
\end{definition}

The dispersion measures how \emph{variable} the compression difficulty is across
source symbols.
If all source symbols are equally difficult to compress, then $V(D) = 0$ and
the convergence to $\RD$ is faster than $1/\sqrt{n}$.
If different symbols have very different compression difficulties, then $V(D)$
is large and the $1/\sqrt{n}$ penalty is more pronounced.

For the $\Ber(p)$ source, since $X$ takes only two values, we have
\begin{equation}
\label{eq:dispersion_bernoulli}
V(D) = p(1-p)\bigl(\jmath_X(1, D) - \jmath_X(0, D)\bigr)^2.
\end{equation}
This is simply the variance of a Bernoulli random variable that takes value
$\jmath_X(0, D)$ with probability $1-p$ and $\jmath_X(1, D)$ with probability $p$.

\begin{remark}
When $p = 1/2$, the source is symmetric: $\jmath_X(0, D) = \jmath_X(1, D)$
for all $D$, and therefore $V(D) = 0$.
Intuitively, every symbol of a fair coin is equally difficult to compress, so
there is no variability in compression difficulty.
This is a somewhat surprising consequence: for the fair Bernoulli source, the
convergence to $\RD$ is faster than $O(1/\sqrt{n})$.
When $V(D) = 0$, the normal approximation~(\ref{eq:normal_approx}) does not
apply (it requires $V(D) > 0$); the $\sqrt{V/n}$ term vanishes and the
$O(\log n / n)$ remainder becomes the dominant correction.
\end{remark}

Figure~\ref{fig:dispersion} shows $V(D)$ as a function of $D$ for several
source biases.

\begin{figure}[htbp]
    \centering
    \includegraphics[width=0.75\textwidth]{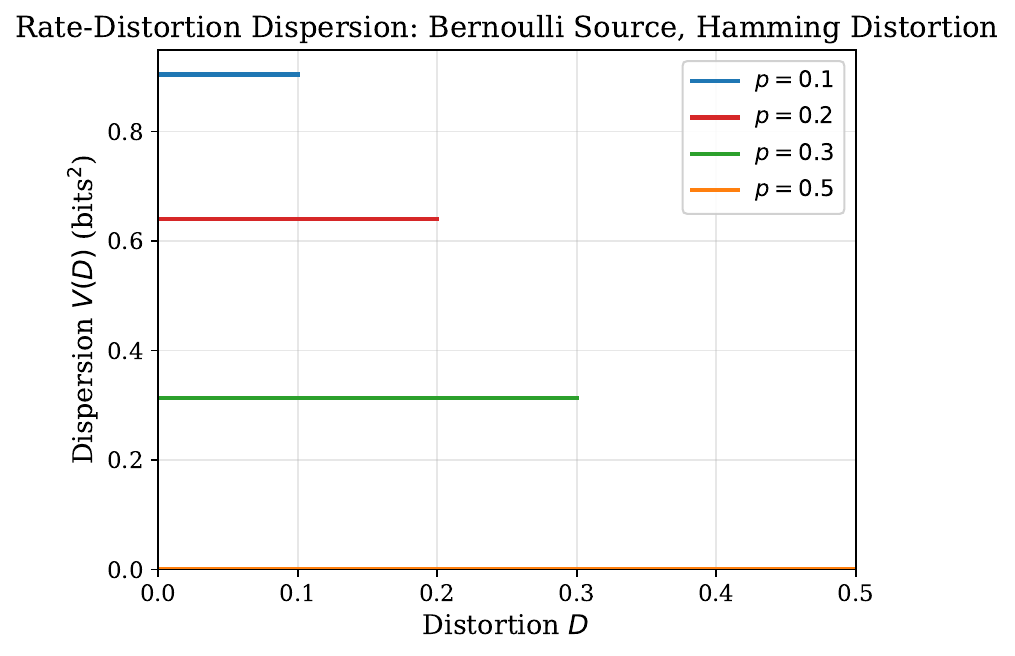}
    \caption{The rate-distortion dispersion $V(D)$ versus distortion $D$ for a
    $\Ber(p)$ source with $p \in \{0.1, 0.2, 0.3, 0.5\}$.
    For $p = 0.5$, the dispersion is identically zero (the source symbols are
    equally difficult to compress).
    For biased sources, the dispersion is largest at intermediate distortion levels.}
    \label{fig:dispersion}
\end{figure}

% ---- 6.6 ----
\subsection{The Normal Approximation}
\label{subsec:normal_approx}

We now state the central result of finite block length rate-distortion theory.
The minimum achievable rate at block length $n$ is characterized by the
following asymptotic expansion.

\begin{theorem}[Normal Approximation~\cite{kostina2012}]
\label{thm:normal_approx}
For a discrete memoryless source with rate-distortion function $\RD$ and
dispersion $V(D) > 0$, the minimum rate at block length $n$ and
excess-distortion probability $\varepsilon \in (0, 1)$ satisfies
\begin{equation}
\label{eq:normal_approx}
R(n, D, \varepsilon) = \RD + \sqrt{\frac{V(D)}{n}}\, \Qinv(\varepsilon)
+ O\!\left(\frac{\log n}{n}\right),
\end{equation}
where $\Qinv(\varepsilon)$ is the inverse of the Gaussian $Q$-function,
$Q(x) = \frac{1}{\sqrt{2\pi}} \int_x^{\infty} e^{-t^2/2}\, dt$.
\end{theorem}

The intuition behind this result comes from the Berry-Esseen central limit
theorem.
The total compression cost for an i.i.d.\ source sequence $X^n$ is
approximately $\sum_{i=1}^n \jmath_X(X_i, D)$, a sum of i.i.d.\ random
variables with mean $n\RD$ and variance $nV(D)$.
By the CLT, this sum is approximately Gaussian, and the excess-distortion
probability translates into a Gaussian tail probability.
The $\Qinv(\varepsilon)$ factor converts the target probability $\varepsilon$
into the number of standard deviations we must accommodate.

Figure~\ref{fig:clt_histogram} illustrates this Gaussian approximation
concretely for a $\Ber(0.3)$ source with $n = 6$ and
$\varepsilon = 0.05$.
For each block of $n = 6$ symbols, the number of ones $k$ ranges from $0$
to $6$, and the per-symbol average $d$-tilted information
$\frac{1}{n}\sum_{i=1}^{n} \jmath_X(X_i, D)$ takes the value
$\frac{k\,\jmath_X(1, D) + (n-k)\,\jmath_X(0, D)}{n}$
with probability $\binom{n}{k}\, p^k (1-p)^{n-k}$.
The blue bars show this exact discrete distribution, while the red curve
shows the Gaussian density $\mathcal{N}\!\bigl(\RD,\, V(D)/n\bigr)$.
The dashed black line marks the Shannon limit $\RD$, and the
dash-dotted green line marks $R(n, D, \varepsilon)$: the minimum rate at
which the Gaussian tail probability (shaded) does not exceed $\varepsilon$.
Comparing the two panels reveals the effect of the distortion target:
at $D = 0.1$ (left), both $\RD$ and $V(D)$ are larger, so the distribution
is wider and the gap $R(n, D, \varepsilon) - \RD$ is more pronounced;
at $D = 0.2$ (right), the distribution is more concentrated and the
finite block length penalty is smaller.

\begin{figure}[htbp]
    \centering
    \includegraphics[width=\textwidth]{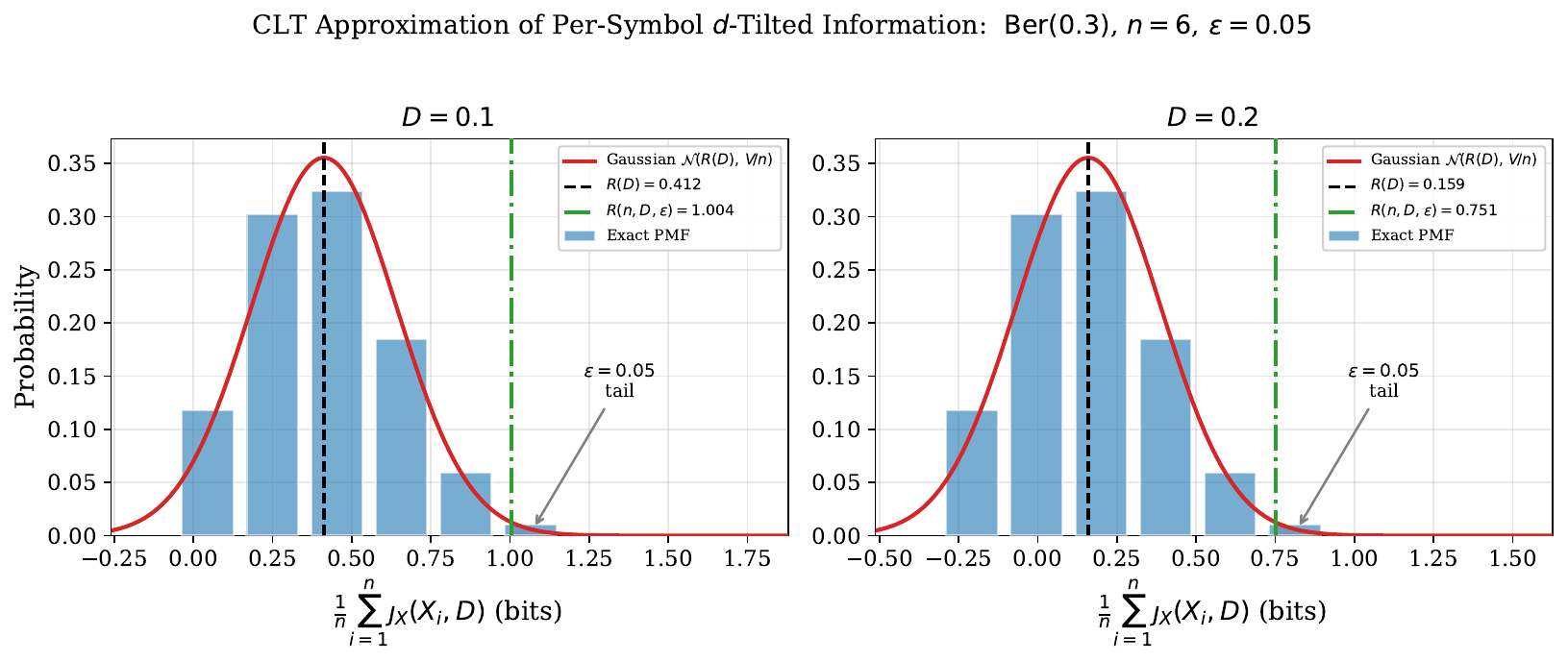}
    \caption{Gaussian approximation of the per-symbol $d$-tilted information
    for a $\Ber(0.3)$ source with $n = 6$ and $\varepsilon = 0.05$, at two
    distortion levels.
    Blue bars show the exact PMF over the $n+1 = 7$ possible values;
    the red curve is the Gaussian $\mathcal{N}(\RD,\, V(D)/n)$.
    The dashed line is $\RD$ and the dash-dotted line is $R(n, D, \varepsilon)$;
    the shaded tail corresponds to the excess-distortion probability
    $\varepsilon$.
    At $D = 0.1$ (left), higher dispersion produces a wider distribution
    and a larger finite block length penalty; at $D = 0.2$ (right), the
    distribution is tighter and the penalty smaller.}
    \label{fig:clt_histogram}
\end{figure}

From an engineering standpoint, Theorem~\ref{thm:normal_approx} provides a
practical design rule: to operate within rate $\RD + \Delta R$ of the Shannon
limit with excess-distortion probability $\varepsilon$, we need a block length of
approximately
\begin{equation}
\label{eq:blocklength_design}
n \approx \frac{V(D) \bigl(\Qinv(\varepsilon)\bigr)^2}{(\Delta R)^2}.
\end{equation}
This expression reveals the fundamental trade-off among block length, rate
overhead, distortion, and reliability.

Figure~\ref{fig:finite_blocklength} shows the achievability bound, converse
bound, and normal approximation for a $\Ber(0.3)$ source.

\begin{figure}[htbp]
    \centering
    \includegraphics[width=0.75\textwidth]{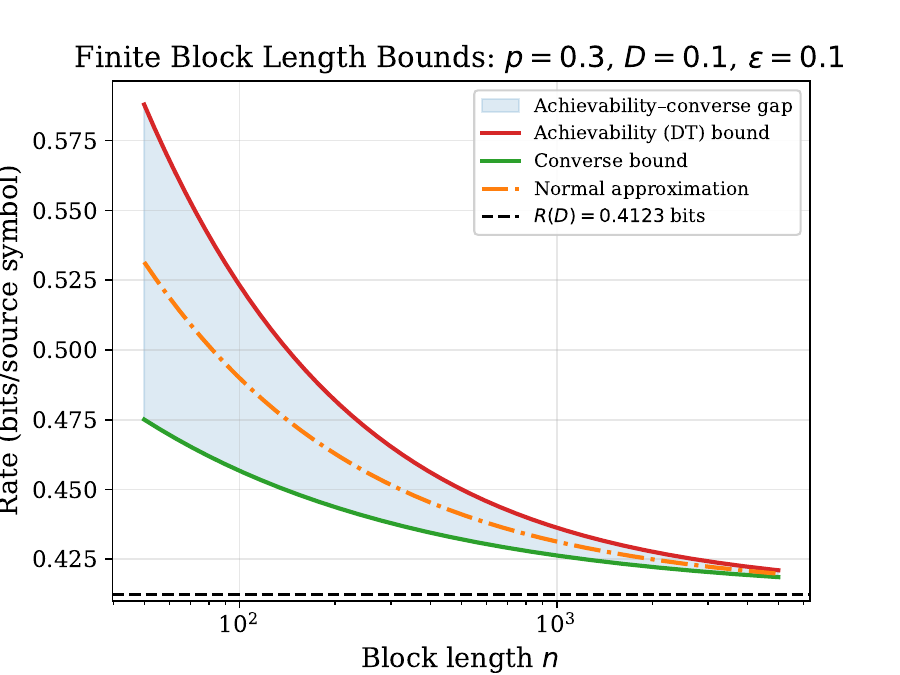}
    \caption{Finite block length bounds for a $\Ber(0.3)$ source with $D = 0.1$
    and $\varepsilon = 0.1$.
    The shaded region between the achievability bound (upper) and converse bound
    (lower) contains the true minimum rate $R(n, D, \varepsilon)$.
    The normal approximation (dashed) lies within this region.
    The horizontal line shows the Shannon limit $\RD$.}
    \label{fig:finite_blocklength}
\end{figure}

% ---- 6.7 ----
\subsection{Historical Note}
\label{subsec:fbl_history}

The study of finite block length performance in information theory dates back
to Strassen~\cite{strassen1962}, who established the $O(1/\sqrt{n})$ refinement
for hypothesis testing.
The channel coding counterpart was developed by Polyanskiy, Poor, and
Verd\'{u}~\cite{polyanskiy2010}, whose work on channel dispersion inspired
the lossy source coding treatment.
Kostina and Verd\'{u}~\cite{kostina2012} established the rate-distortion
dispersion and the normal approximation~(\ref{eq:normal_approx}) for general
discrete memoryless sources, building on contributions by
Ingber and Kochman~\cite{ingber2011}.
The $d$-tilted information framework provides a unified language for
second-order information theory.

% ========================================================================
% SECTION 7: NUMERICAL EXPLORATIONS
% ========================================================================
\section{Numerical Explorations}
\label{sec:numerical}

Every figure in this tutorial was generated programmatically by accompanying
Python scripts.
The scripts, along with the \LaTeX{} source and all generated figures, are
publicly available at:\\
\texttt{https://github.com/anrgusc/rate-distortion-bernoulli-finite}.

\subsection{Computational Tools and Figures}

We provide six Python scripts in the \texttt{scripts/} directory:
\begin{itemize}
    \item \texttt{rate\_distortion.py}: Computes and plots the binary entropy
          function (Figure~\ref{fig:binary_entropy}) and rate-distortion curves
          (Figure~\ref{fig:rd_curves}).
    \item \texttt{blahut\_arimoto.py}: Implements the Blahut-Arimoto algorithm
          with convergence tracking (Figure~\ref{fig:ba_convergence}) and
          validates it against the closed-form solution
          (Figure~\ref{fig:ba_vs_closedform}).
    \item \texttt{dispersion.py}: Computes the $d$-tilted information
          (Figure~\ref{fig:dtilted}) and rate-distortion dispersion
          (Figure~\ref{fig:dispersion}).
          This script also runs a numerical verification that
          $\E[\jmath_X(X,D)] = \RD$ to machine precision.
    \item \texttt{clt\_histogram.py}: Illustrates the CLT approximation
          underlying the normal approximation by plotting the exact PMF
          of the per-symbol $d$-tilted information alongside the Gaussian
          density, for two distortion levels
          (Figure~\ref{fig:clt_histogram}).
    \item \texttt{finite\_blocklength.py}: Computes the normal approximation,
          achievability/converse bounds (Figure~\ref{fig:finite_blocklength}),
          and the rate versus block length curves
          (Figure~\ref{fig:rate_vs_blocklength}).
    \item \texttt{generate\_all\_figures.py}: Master script that runs all of the
          above and generates every figure with consistent styling.
\end{itemize}

To regenerate all figures from scratch, run
\texttt{python scripts/generate\_all\_figures.py} from the repository root.
The only dependencies are NumPy, SciPy, and Matplotlib (see
\texttt{requirements.txt}).

\subsection{Comprehensive Comparison}

Figure~\ref{fig:comprehensive} brings together the asymptotic rate-distortion
function and its finite block length refinements in a single plot.
For a $\Ber(0.3)$ source with excess-distortion probability $\varepsilon = 0.1$,
we show the achievable rate as a function of distortion $D$ for several block
lengths $n \in \{100, 500, 2000\}$.
As $n$ increases, the finite block length curves converge uniformly to the
Shannon limit $\RD$.
The gap is most pronounced at small distortions, where the dispersion $V(D)$ is
larger.

\begin{figure}[htbp]
    \centering
    \includegraphics[width=0.75\textwidth]{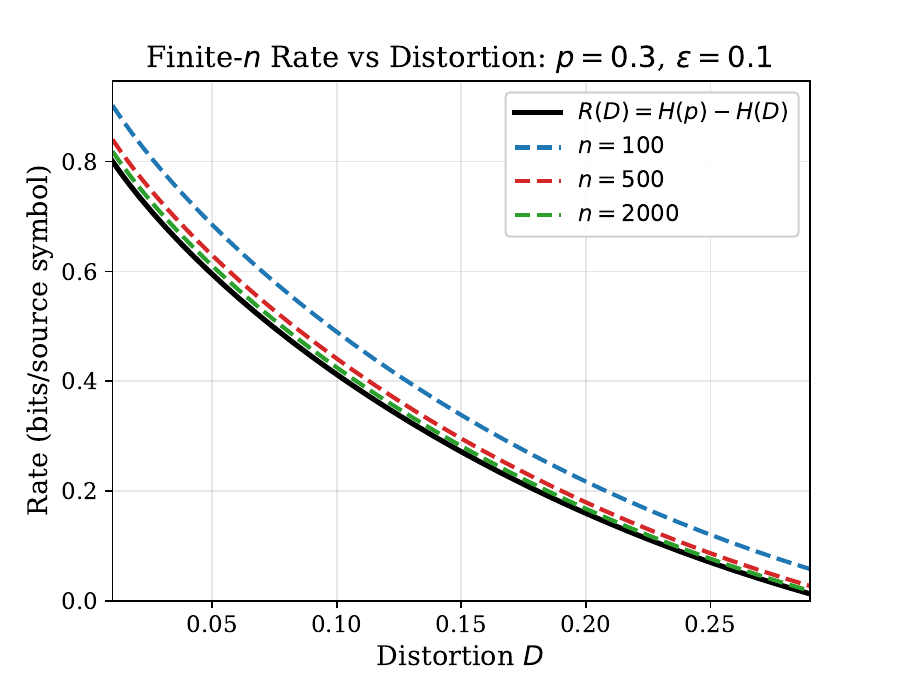}
    \caption{The rate-distortion function $\RD$ (solid black) and the normal
    approximation $R(n, D, \varepsilon)$ for block lengths $n \in \{100, 500, 2000\}$,
    with $p = 0.3$ and $\varepsilon = 0.1$.
    The finite block length penalty is largest at small $D$ (high rate regime)
    and vanishes as $D \to \min(p, 1-p)$.}
    \label{fig:comprehensive}
\end{figure}

\subsection{Blahut-Arimoto Convergence}

The Blahut-Arimoto algorithm exhibits geometric convergence, as shown in
Figure~\ref{fig:ba_convergence}.
The convergence rate depends on the slope parameter $s$: larger $s$
(corresponding to smaller target distortions) leads to faster convergence in
terms of iteration count.
For all tested parameters, the algorithm reaches a relative error below $10^{-10}$
within $50$ iterations.

% ========================================================================
% SECTION 8: CONCLUSION
% ========================================================================
\section{Conclusion}
\label{sec:conclusion}

In this tutorial, we have presented a self-contained development of
rate-distortion theory for the $\Ber(p)$ source with Hamming distortion,
progressing from Shannon's asymptotic limit to the finite block length regime.

The key takeaways are threefold.
First, the rate-distortion function $\RD = \Hp - \HD$ provides a clean and
interpretable limit on lossy compression: the minimum rate is the source
entropy minus the noise entropy.
Second, the Blahut-Arimoto algorithm provides a reliable computational tool
for evaluating rate-distortion functions, even in settings where closed-form
solutions are unavailable.
Third, the finite block length theory, centered on the $d$-tilted information
$\jmath_X(x, D)$ and the dispersion $V(D)$, provides a precise characterization
of the penalty for operating at practical block lengths.
The normal approximation $R(n, D, \varepsilon) \approx \RD + \sqrt{V(D)/n}\, \Qinv(\varepsilon)$
is both elegant in form and useful in practice, offering a direct design rule
for system engineers.

% ========================================================================
% ACKNOWLEDGEMENT
% ========================================================================
\section*{Acknowledgement}

The development of this tutorial and the associated Python code has involved
the use of several AI tools/agents including Claude Code, ChatGPT, and Gemini.
The human author accepts full responsibility for its contents.

% ========================================================================
% BIBLIOGRAPHY
% ========================================================================
\bibliographystyle{IEEEtran}
\bibliography{references}

\end{document}